\documentclass[num-refs,serif]{wiley-article}

\usepackage{graphicx} 
\usepackage{hyperref}
\usepackage{amsmath}
\usepackage{arydshln} 
\usepackage{setspace}
\usepackage{caption}
\usepackage{verbatim}
\usepackage{bm}
\usepackage{relsize}
\usepackage[normalem]{ulem} 
\usepackage{tikz}
\usetikzlibrary{arrows.meta}
\usepackage{float} 
\usepackage{booktabs}
\usepackage{multirow}

\newtheorem{Theorem}{Theorem}

\newtheorem{Lemma}{Lemma}

\papertype{Original Article}
\paperfield{Journal Section}

\title{Saturation in G: simple \& robust causal inference in cluster randomized trials with informative cluster sizes}

\author[1,2,3\authfn{1}]{Kenneth M. Lee}
\author[1,2,3]{Michael O. Harhay}
\author[4,5]{Fan Li}

\affil[1]{Department of Biostatistics, Epidemiology and Informatics, University of Pennsylvania, Philadelphia, PA, USA}
\affil[2]{Center for Clinical Trials Innovation, Department of Biostatistics, Epidemiology, Informatics, University of Pennsylvania, Philadelphia, Pennsylvania, USA}
\affil[3]{Clinical Trials Methods and Outcomes Lab, Palliative and Advanced Illness Research (PAIR) Center, Perelman School of Medicine, University of Pennsylvania, Philadelphia, PA, USA}
\affil[4]{Department of Biostatistics, Yale School of Public Health, New Haven, CT, USA}
\affil[5]{Center for Methods in Implementation and Prevention Science, Yale School of Public Health, New Haven, CT, USA}

\corraddress{\authfn{1} Department of Biostatistics, Epidemiology and Informatics, University of Pennsylvania School of Medicine, 3600 Civic Center Boulevard, Philadelphia, PA 19104}
\corremail{kenneth.lee@pennmedicine.upenn.edu}

\fundinginfo{Funding information is included at the end of the article.}

\runningauthor{Lee et al.}

\begin{document}

\begin{frontmatter}
\maketitle

\begin{abstract}
\small
Cluster randomized trials (CRTs) can exhibit informative cluster sizes (ICS) where cluster size is associated with outcomes and/or treatment effects. 
Under ICS, the individual and cluster-average treatment effects (iATE, cATE) can diverge, and the conventional linear mixed-effects model (LMM) and generalized estimating equation (GEE) with an exchangeable working correlation can produce data-dependent weighted contrasts that are not consistent for either estimand.
In these settings with ICS, we propose easy to implement ``cluster-size saturated models with g-computation'' (CS-g), which employ a simple two-step adjustment to standard practice: (1.) augment the appropriately weighted working LMM or GEE with a saturated continuous cluster-size main effect and treatment $\times$ cluster-size interaction, and (2.) apply g-computation to target an interpretable marginal estimand.
We prove that the appropriately weighted cluster-size saturated LMM with g-computation and more general cluster-size saturated GEE with g-computation can consistently target the iATE and cATE, among a broad class of interpretable estimands, while allowing for ICS. 
Crucially, this consistency holds under arbitrary misspecification of other model components, including the functional form of the saturated cluster-size terms.
Furthermore, we demonstrate exact finite-sample equivalence between these consistent CS-g estimators and their model-robust standardization counterparts.
Across simulations with continuous and binary outcomes, the proposed CS-g estimators were unbiased,  more efficient than other consistent estimators, and returned greater power to detect ICS.
A re-analysis of the PPACT P-CRT further illustrates the approach. 
Altogether, CS-g offers a simple, robust, and efficient route to target interpretable marginal effects in P-CRTs with ICS.

\keywords{cluster randomized trial, informative cluster sizes, saturation, g-computation, standardization, robust}
\end{abstract}
\end{frontmatter}

\section{Introduction}

Cluster randomized trials (CRTs) refer to the collection of study designs where randomization is carried out at the cluster level (e.g., hospital, clinic, or worksite level), with outcome measurements often collected at the individual participant level \cite{hayes_cluster_2017}. 
Among the many existing CRT design variations, the standard parallel cluster randomized trial (P-CRT) is the simplest, with clusters randomized to implement either the treatment or control over the entire trial duration.

In a P-CRT, researchers can define two target estimands with particularly natural interpretations: the cluster-average treatment effect (cATE) and the individual-average treatment effect (iATE) \cite{kahan_informative_2023,kahan_demystifying_2024,kahan_crt-estimands_2026}. 
Briefly, the cATE (sometimes also referred to as the “unit average treatment effect” or UATE \cite{imai_essential_2009}) is the average treatment effect giving all clusters equal weight, with individuals pooled across their corresponding cluster, and can be of interest when studying interventions designed for implementation at the cluster level.
Inverse cluster size weights can then be specified to ensure that all clusters contribute equally regardless of cluster size \cite{williamson_marginal_2003}.
The iATE (sometimes also referred to as the ``participant average treatment effect'') is the average treatment effect giving all individuals equal weight, mimicking what would typically be targeted in an individually-randomized trial, and can be of relevance when studying individual-level interventions that are cluster randomized due to logistical or administrative considerations.
Notably, these two estimands can differ in the presence of heterogeneous treatment effects that vary according to cluster size, also referred to as ``informative cluster sizes'' (ICS) \cite{kahan_informative_2023,kahan_demystifying_2024, williamson_marginal_2003,bugni_inference_2024}. 

Historically, P-CRTs are often analyzed with an exchangeable correlation structure specified through a mixed-effects model or generalized estimating equation (GEE) framework to account for within-cluster correlation \cite{hayes_cluster_2017}. 
With an identity link, the corresponding point estimators from a linear mixed-effects model with a random intercept and a GEE with an exchangeable working correlation will coincide \cite{gardiner_fixed_2009,hubbard_gee_2010}.
This mixed-effects model is often preferred for modeling the assumed underlying data generating process, and is the best linear unbiased estimator (BLUE) when the model is correctly specified \cite{girling_statistical_2016}.
These models can also be model-robust when certain assumptions are met, including an assumption of non-informative sampling or arm-specific random sampling to functionally avoid ICS \cite{wang_how_2024,wang_mixed-model_2026,wang_model-robust_2024}.

Despite their widespread implementation, previous work has indicated that in the presence of ICS, treatment effect coefficients from such a linear mixed-effects model produce model-based estimators that have undesirable data-dependent weights and are neither consistent for the iATE nor cATE in the analysis of P-CRTs \cite{wang_two_2022}.
The interpretation of similar model coefficients can be even less transparent for nonlinear links (e.g., logit-links), where generalized linear mixed-effects model coefficients target conditional estimands, unlike corresponding GEEs that target marginal estimands.
In contrast, work in P-CRTs \cite{wang_two_2022}, along with corresponding studies in other CRT designs \cite{lee_how_2024,lee_what_2025}, have found that appropriately weighted independence estimating equations (IEE; equivalent to ordinary least squares in a linear model) yield consistent estimators for the iATE and cATE estimands under ICS.
However, the implementation of an IEE can face resistance, with the mixed-effects model often still preferred for the previously stated reasons.

While the model-based estimators resulting from the basic mixed-effects model or GEE with an exchangeable working correlation are neither consistent for the marginal iATE nor cATE estimands in the presence of ICS, recent work has derived a ``model-robust standardization'' (MRS) approach that can be used alongside these basic models to produce consistent estimators \cite{li_model-robust_2025}.
MRS achieves this consistency by extending g-computation (a simple procedure using model predictions to target marginal estimands; also referred to as ``marginal standardization'' \cite{rosenbaum_model-based_1987}) with an additional weighted cluster-level residual term \cite{li_model-robust_2025}.

In this work, we instead propose a simple parametric model adjustment to standard mixed-effects models and GEEs with an exchangeable working correlation that allows researchers to target the iATE and cATE estimands in P-CRTs, despite the presence of ICS.
This approach (1.) implements a continuous cluster-size main effect and cluster-size $\times$ treatment interaction in the model, and (2.) uses g-computation to target an interpretable marginal estimand.
We refer to this approach as ``cluster-size saturated models with g-computation'' (CS-g), and will explicitly focus on linear mixed-effects models (CSLMM-g, CSLMMw-g) and GEEs (CSGEE-g, CSGEEw-g) without or with inverse cluster-size weights.
In contrast to MRS, the proposed CS-g approach instead uses model adjustment within the simpler g-computation framework to achieve robust consistency.
The CS-g modeling approach was previously mentioned by Hooper et al. \cite{hooper_wood_2026}, which recommended careful modeling of the treatment $\times$ cluster-size interactions to accurately capture the assumed underlying treatment effect heterogeneity resulting from ICS.
In contrast, we will provide a model-robust perspective to these estimators and prove that in P-CRTs with ICS, the CSLMM-g \& CSLMMw-g and the more general CSGEE-g \& CSGEEw-g estimators can robustly target clearly defined marginal individual and cluster-level estimands, despite arbitrary misspecification of the correlation and covariate structures, including the functional form of the cluster-size main effect and treatment interaction terms.

We first outline some standard assumptions while allowing for ICS (Section \ref{sect:assumptions}) and define the iATE and cATE among a broad class of individual and cluster-level estimands (Section \ref{sect:estimands}).
We then describe the CS-g modeling approach (Section \ref{sect:CS-g}) and prove the model-robust consistency of the appropriately weighted CSLMM-g \& CSLMMw-g and more general CSGEE-g \& CSGEEw-g estimators for the iATE and cATE in P-CRTs with ICS (Section \ref{sect:proof}).
Furthermore, we prove that consistent estimators from these CS-g models share exact finite sample equivalence to corresponding estimators from cluster-size saturated models with MRS (Section \ref{sect:g=MRS}).
We then empirically assess the performance of these aforementioned consistent CS-g models,
alongside the CSGLMM-g \& CSGLMMw-g, by simulation (Section \ref{sect:sim}).
We include a re-analysis of a P-CRT case study example (Section \ref{sect:case}) and end with some concluding remarks (Section \ref{sect:discussion}).

\section{Assumptions}
\label{sect:assumptions}

Consider a P-CRT with $m$ clusters, where each cluster $i \in \{1,...,m\}$ contains $N_i$ individuals in its source population (Figure \ref{fig:p_crt_design}).
We assume $N_i$ can vary by clusters and takes values in a bounded subset of positive integers. For simplicity, we assume that all $N_i$ individuals in the cluster source population are included in the P-CRT study.
For each individual $k \in \{1,...,N_i\}$ in cluster $i$, we define $Y_{ik}$ as their outcome, $\bm{X}_{ik}$ as their vector of covariates, and $Z_i=1$ or $Z_i=0$ if cluster $i$ receives treatment or control, respectively.
We use the potential outcomes framework to eventually define treatment effect estimands; let $Y_{ik}(1)$ denote the potential outcome of individual $k$ in cluster $i$ had the cluster received treatment, otherwise, $Y_{ik}(0)$ denotes the untreated potential outcome.

\begin{figure}[ht!]
\setlength{\unitlength}{0.10in} 
\centering 
\begin{picture}(42,13)(-10,6) 
\setlength\fboxsep{0pt}

\put(-4,17){Cluster $i=1$}
\put(3,16.5){\colorbox{gray!40}{\framebox(20,1.5){$\{\bm{O}_1=Y_{1k},\bm{X}_{1k},Z_1=1 : k=1,\dots,N_1\}$}}}

\put(-4,15){Cluster $i=2$}
\put(3,14.5){\colorbox{gray!40}{\framebox(20,1.5){$\{\bm{O}_2=Y_{2k},\bm{X}_{2k},Z_2=1 : k=1,\dots,N_2\}$}}}

\put(-4,13){Cluster $i=3$}
\put(3,12.5){\colorbox{gray!40}{\framebox(20,1.5){$\{\bm{O}_3=Y_{3k},\bm{X}_{3k},Z_3=1 : k=1,\dots,N_3\}$}}}

\put(-4,11){Cluster $i=4$}
\put(3,10.5){\framebox(20,1.5){$\{\bm{O}_4=Y_{4k},\bm{X}_{4k},Z_4=0 : k=1,\dots,N_4\}$}}

\put(-4,9){Cluster $i=5$}
\put(3,8.5){\framebox(20,1.5){$\{\bm{O}_5=Y_{5k},\bm{X}_{5k},Z_5=0 : k=1,\dots,N_5\}$}}

\put(-4,7){Cluster $i=6$}
\put(3,6.5){\framebox(20,1.5){$\{\bm{O}_6=Y_{6k},\bm{X}_{6k},Z_6=0 : k=1,\dots,N_6\}$}}

\end{picture}
\caption{Design of a P-CRT with clusters $i$ assigned to receive either treatment ($Z_i=1$) or control ($Z_i=0$), producing the observed data ($\bm{O}_i$).}
\label{fig:p_crt_design}
\end{figure}

The observed data for each cluster are denoted as $\bm{O}_i=\{Y_{ik}, \bm{X}_{ik}, Z_i : k=1,...,N_i\}$ (Figure \ref{fig:p_crt_design}).
The complete, but not fully observed, data vector for each cluster $i$ is then denoted as $\bm{W}_i = \{(Y_{ik}(0), Y_{ik}(1), \allowbreak \bm{X}_{ik}, Z_i, N_i) : k=1,...,N_i\}$. 
To proceed, we outline the following assumptions on $\bm{W}_i$ to enable the subsequent estimand definitions and theorem proofs.

\noindent \textbf{A1.} (\textit{Super-population sampling}) Data vectors $\{\bm{W}_i, i=1,...,m\}$  are independent and identically distributed draws from a population distribution $\mathcal{P}$ with finite second moments.
Furthermore, within each cluster $i$, the data vectors $\{(Y_{ik}(0), Y_{ik}(1)), \bm{X}_{ik}\}$ for $k=1,...,N_i$ are identically distributed given the source population size $N_i$.

\noindent \textbf{A2.} (\textit{Cluster randomization}) Treatment assignment $Z_i$ is independent of all other random variables in $\bm{W}_i$, and $P(Z_i=1)=E[Z_i]=\pi \in (0,1)$.

We illustrate the super-population sampling scheme (Assumption A1) and subsequent cluster randomization (Assumption A2) in Figure \ref{fig:sampling}.
Assumption A1 is typical for causal inference under a super-population framework. Although the dimension of $\bm{W}_i$ varies by $N_i$, the data can be generated via the mixture model $\mathcal{P} = \mathcal{P}^{\bm{W}|N} \times \mathcal{P}^N$, where we first draw cluster size $N_i \sim \mathcal{P}^N$, then draw within-cluster data $\bm{W}_i \sim \mathcal{P}^{\bm{W}|N}$ given $N_i$.
In addition, this assumption requires that the individual-level complete data vector has the same expectation conditional on $N_i$; this allows us to construct estimands based on marginal expectations of individual-level potential outcomes.
Importantly, although A1 assumes between-cluster independence, it allows for arbitrary within-cluster correlation structures among potential outcomes and covariates.
For notational convenience, we omit the subscript $i$ when taking expectation with respect to population distribution $\mathcal{P}$ (Assumption A1). For example, $E[f(\bm{O_i})]$, where $\bm{O}_i$ is the observed data of cluster $i$ and $f$ is an arbitrary measurable function, is simplified as $E[f(\bm{O})]$, where $\bm{O}$ represents the random variable sampled from $\mathcal{P}$.
Likewise, $(Y_{ik}, N_{i})$ are denoted as $(Y_{.k}, N)$ when taking expectations. This notational simplification is used throughout the manuscript.
Assumption A2 holds by design for cluster randomized trials.

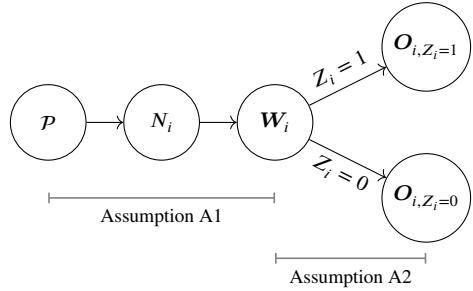
\begin{figure}[h]
    \centering
    \begin{tikzpicture}[
    node distance=2cm,
    main/.style={circle, draw, minimum size=1cm},
    lbl/.style={font=\small, fill=white, inner sep=2pt, align=center}, 
    every edge/.style={draw, -Stealth}
]
    \node[main] (P) at (-3, 0) {$\mathcal{P}$};
    \node[main] (N) at (-1.5,0) {$N_i$};
    \node[main] (W) at (0,0) {$\bm{W}_i$};
    \node[main] (O0) at (2,-1) {$\bm{O}_{i, Z_i=0}$};
    \node[main] (O1) at (2,1) {$\bm{O}_{i,Z_i=1}$};
    
    \draw[->] (P) -- (N);
    \draw[->] (N) -- (W);
    \draw[->] (W) -- node[above, sloped] {$Z_i=1$} (O1);
    \draw[->] (W) -- node[below, sloped] {$Z_i=0$} (O0);

    \draw[{Bar[width=4pt]}-{Bar[width=4pt]}, semithick, gray] (-3, -1.0) -- (0, -1.0) node[midway, below=2pt, lbl, text=black]{Assumption A1};
    \draw[{Bar[width=4pt]}-{Bar[width=4pt]}, semithick, gray] (0, -1.8) -- (2, -1.8) node[midway, below=2pt, lbl, text=black] {Assumption A2};
    \end{tikzpicture}
    \caption{Illustration of the described super-population sampling scheme (Assumption A1) and subsequent cluster randomization (Assumption A2) while allowing for ICS.}
    \label{fig:sampling}
\end{figure}

Unlike similar publications that establish the model-robustness of mixed-effects models in CRTs by making additional assumptions to functionally avoid ICS \cite{wang_how_2024,wang_mixed-model_2026,wang_model-robust_2024}, we explicitly prove the consistency of our proposed modeling approach while explicitly allowing for ICS.

\section{Treatment Effect Estimands and Informative Cluster Sizes}
\label{sect:estimands}

Let $Y_{ik}(z)$ denote the potential outcome of individual $k$ in cluster $i$, with $z \in \{0,1\}$. Then, $Y_{ik}(0)$ and $Y_{ik}(1)$ denote the untreated and treated potential outcomes, respectively.
We connect the observed outcome $Y_{ik}$ and potential outcomes via
\[
Y_{ik} =I\{Z_i=1\}Y_{ik}(1) + I\{Z_i=0\}Y_{ik}(0) \,.
\] 
The cluster-specific vector of observed and potential outcomes is then
\[
\begin{split}
    \bm{Y}_{i} &= (Y_{i1},...,Y_{iN_i})^\top \in \mathbb{R}^{N_i} \,, \\
    \bm{Y}_{i}(z) &= (Y_{i1}(z),...,Y_{iN_i}(z))^\top \in \mathbb{R}^{N_i} \,. \\ 
\end{split}
\]

With informative cluster sizes (ICS), cluster size $N_i$ is correlated with potential outcomes $\{Y_{ik}(0),Y_{ik}(1)\}$, specifically being dependent on (I.) baseline outcomes ($Y_{ik}(0)$) and/or (II.) treatment effects ($Y_{ik}(1)-Y_{ik}(0)$) \cite{kahan_demystifying_2024}.
We can define the individual-average treatment effect (iATE) and cluster-average treatment effect (cATE) estimands, with expectations taken over the distribution of clusters ($E[f(\bm{W}_i)] = \int f(\bm{w}) d\mathcal{P}(\bm{w})$ for any integrable function $f$)
\[
\begin{split}
    \Delta_{iATE} &= \frac{E\left[\sum_{k=1}^N(Y_{.k}(1)-Y_{.k}(0))\right]}{E\left[N\right]} = \frac{E\left[\bm{1}_N^\top (\bm{Y}(1)-\bm{Y}(0))\right]}{E\left[N\right]}  \,, \\
    \Delta_{cATE} &= E\left[\frac{1}{N}\sum_{k=1}^N(Y_{.k}(1)-Y_{.k}(0))\right] = E\left[\frac{1}{N}\bm{1}_N^\top (\bm{Y}(1)-\bm{Y}(0))\right] \,.
\end{split}
\]
These estimand definitions are model-free and differ in the presence of ICS.II. 
The iATE estimand treats each individual as the unit of inference, giving equal weight to all participants regardless of which cluster they belong to. If researchers are primarily interested in an effect on the average individual across the population of participants, the iATE estimand may be more appropriate. This may be suitable with participant-level outcome measures such as mortality, quality of life, etc.
The cATE estimand treats each cluster as the unit of inference, giving equal weight to all clusters regardless of their size. This can be achieved by applying inverse cluster-size weights to the previously described individual-level comparisons or by analyzing cluster-level summaries. If researchers are primarily interested in an effect on the average cluster across the population of clusters, the cATE estimand may be more appropriate. This may be suitable when outcomes capture cluster behavior.
More discussion around these two estimands can be found in some recent publications \cite{kahan_estimands_2023,kahan_demystifying_2024,kahan_crt-estimands_2026}.

The iATE and cATE belong to a broad class of
weighted average treatment effect estimands defined here as
\[
    \Delta_{\bm{\lambda} A}=h\left\{ \frac{E\left[(\lambda/N)\bm{1}_N^\top \bm{Y}(1)\right]}{E\left[\lambda\right]} , \frac{E\left[(\lambda/N)\bm{1}_N^\top \bm{Y}(0)\right]}{ E\left[\lambda\right]} \right\} \,.
\]
Then, $\lambda=N$ returns the individual-average estimand and $\lambda=1$ returns the cluster-average estimand
\[
    \Delta_{iA} = h\left\{ \frac{E\left[\bm{1}_N^\top \bm{Y}(1)\right]}{E\left[N\right]} , \frac{E\left[\bm{1}_N^\top \bm{Y}(0)\right]}{E\left[N\right]} \right\}, \,\,\,\, 
    \Delta_{cA} = h\left\{ E\left[\frac{1}{N}\bm{1}_N^\top \bm{Y}(1)\right] , E\left[\frac{1}{N}\bm{1}_N^\top \bm{Y}(0)\right] \right\} \,,
\]
and $h$ is a pre-specified function determining the contrasts and scale of the effect measure. Then, $h\{x,y\}=x-y$ yields the previously described difference estimands (iATE, cATE), $h\{x,y\}=x/y$ yields the marginal risk ratio estimands (iARR, cARR), and $h\{x,y\}=\{x/(1-x)\}/\{y/(1-y)\}$ yields the marginal odds ratio estimands (iAOR, cAOR).

\section{CS-g: Cluster-size Saturated Models with g-computation}
\label{sect:CS-g}

First, consider the \textit{basic} working linear mixed-effects model for individual $k$ in cluster $i$
\[
    Y_{ik} = \beta_0 + I\{Z_i=1\} \beta_Z + \bm{X}_{ik} \bm{\beta_X} + \alpha_i + \epsilon_{ik}
\]
which can be summarized in vector form across cluster $i$ as
\begin{equation}
    \label{eq:ME_basic}
    \bm{Y}_i =  \bm{1}_{N_i}\beta_0 + I\{Z_i=1\} \bm{1}_{N_i} \beta_Z + \bm{X}_i \bm{\beta_X} + \alpha_i \bm{1}_{N_i} + \bm{\epsilon}_i \,,
\end{equation}
where recall $\bm{Y}_{i} = (Y_{i1},...,Y_{iN_i})^\top \in \mathbb{R}^{N_i}$ and $\bm{1}_{N_i}$ is a $N_i$-dimensional vector of ones. Then $\beta_0$ is the model intercept, $\beta_Z$ is the treatment effect coefficient, $\bm{\beta_X} \in \mathbb{R}^p$ is the vector of $p$ covariates (individual and/or cluster-level), $\bm{X}_i=(\bm{X}_{i1},...,\bm{X}_{iN_i})^\top \in \mathbb{R}^{N_i \times p}$ indicate the $p$ covariates, $\alpha_i \sim N(0,\tau^2)$ is the cluster random intercept inducing within-cluster correlation, and $\bm{\epsilon}_i = (\epsilon_{i1}, ..., \epsilon_{iN_i})^\top$ with $\epsilon_{ik} \sim N(0,\sigma^2)$ are the residuals.

Crucially, the treatment effect estimator $\beta_Z$ from such a model and its inverse cluster-size weighted counterpart have been previously proven to be inconsistent for the iATE and cATE estimands in P-CRTs with ICS.II \cite{wang_two_2022}.

\subsection{Cluster-size Saturated Models}

\label{sect:CSLMM_CSGEE}

Building on the basic model described in Equation (\ref{eq:ME_basic}), the \textit{cluster-size saturated linear mixed-effects model} (CSLMM) for individual $k$ in cluster $i$ can then be specified as
\[
    Y_{ik} = \beta_0 + I\{Z_i=1\} \beta_Z +I\{Z_i=1\} N_i \beta_{ZN} +  N_i \beta_{N} +  \bm{X}_{ik} \bm{\beta_X} + \alpha_i + \epsilon_{ik}
\]
and can be summarized in vector form across cluster $i$ as
\begin{equation}
\label{eq:ME_sat}
    \bm{Y}_i =  \bm{1}_{N_i}\beta_0 + I\{Z_i=1\} \bm{1}_{N_i} \beta_Z + I\{Z_i=1\} N_i \bm{1}_{N_i} \beta_{ZN} + N_i \bm{1}_{N_i} \beta_{N}  + \bm{X}_i \bm{\beta_X} + \alpha_i \bm{1}_{N_i} + \bm{\epsilon}_i \,,
\end{equation}
where cluster-size ($N_i$) is included as a continuous covariate; then $\beta_{ZN}$ is the treatment $\times$ cluster-size interaction term and $\beta_N$ is the cluster-size main effect.
Then, the cluster-size saturated working model is
\begin{equation}
\label{eq:ME_sat_dist}
    \bm{Y}_i|  Z_i, \bm{X}_i, N_i 
    \sim N(\bm{Q}_i\bm{\beta}, \bm{\Sigma}_i) \,,
\end{equation}
where $\bm{Q}_i=(\bm{1}_{N_i}, I\{Z_i=1\} \bm{1}_{N_i}, I\{Z_i=1\} N_i \bm{1}_{N_i}, N_i \bm{1}_{N_i}, \bm{X}_i) \in \mathbb{R}^{N_i \times (4+p)}$, 
$\bm{\beta}=(\beta_0, \beta_Z, \beta_{ZN}, \allowbreak \beta_N ,\bm{\beta}_X^{\top})^{\top} \in \mathbb{R}^{4+p}$, and $\bm{\Sigma}_i = \tau^2 (\bm{1}_{N_i}\bm{1}_{N_i}^\top) + \sigma^2 \textbf{I}_{N_i} \in \mathbb{R}^{N_i \times N_i}$.

The more general \textit{cluster-size saturated GEE} (CSGEE) can be specified with the following marginal mean model for individual-level data
\[
\begin{split}
    E\left[Y_{ik} | \bm{Q}_{ik}\right] &=  g^{-1}\left( \beta_0 + I\{Z_i=1\} \beta_Z + I\{Z_i=1\} N_i \beta_{ZN} + N_i \beta_{N}  + \bm{X}_{ik} \bm{\beta_X} \right) \\
    &= g^{-1}\left( \bm{Q}_{ik}\bm{\beta}  \right) \\
    &= \mu_{ik}
\end{split}
\]
which we summarize in vector form across cluster $i$ as
\begin{equation}
\label{eq:GEE_sat}
\begin{split}
    E[\bm{Y}_i | \bm{Q}_i] &=  g^{-1}\left(\bm{1}_{N_i}\beta_0 + I\{Z_i=1\} \bm{1}_{N_i} \beta_Z + I\{Z_i=1\} N_i \bm{1}_{N_i} \beta_{ZN} + N_i \bm{1}_{N_i} \beta_{N}  + \bm{X}_i \bm{\beta_X}\right) \\
    &= g^{-1}(\bm{Q}_i \bm{\beta}) \\
    &= \bm{\mu}_i \,
    \end{split}
\end{equation}
where $g$ is the link function, which we assume to be canonical (e.g., logit for binary outcomes, log for count outcomes, identity for continuous outcomes) and $\bm{Q}_{ik}=(1, I\{Z_i=1\} , I\{Z_i=1\} N_i, N_i, \bm{X}_{ik}) \in \mathbb{R}^{1 \times (4+p)}$, and $\bm{Q}_i = (\bm{Q}_{ik})_{k:1,...,N_i} \in \mathbb{R}^{N_i \times (4+p)}$.
Then the vector of regression coefficients $\bm{\beta}$ is estimated by $\hat{\bm{\beta}}$, being the solution to the estimating equation
\begin{equation}
\label{eq:GEE}
    \sum_{i=1}^{m} \bm{U}_i^\top \bm{\mathcal{Z}}_i^{-1/2} \bm{R}_i^{-1} \bm{\mathcal{Z}}_i^{-1/2} (\bm{Y}_i - \bm{\mu}_i) = 0 \,,
\end{equation}
where $\bm{\mu}_i = g^{-1}(\bm{Q}_{i} \bm{\beta}) = \left(g^{-1}(\bm{Q}_{ik} \bm{\beta})\right)_{k=1,...,N_i} = (\mu_{ik})_{k=1,...,N_i} \in \mathbb{R}^{N_i}$ is the mean function vector for all observed individuals in cluster $i$, 
$\bm{U}_i = \frac{d\bm{\mu}_i}{d\bm{\beta}}$ is the derivative matrix, $\bm{\mathcal{Z}}_i = \text{diag}\{v(Y_{ik}) : k=1,...,N_i\}$ is the diagonal matrix of the natural variance functions $v(Y_{ik})$, and $\bm{R}_i$ is the working correlation structure for the observed outcomes in cluster $i$.
With the GEE, we generally consider the exchangeable working correlation structure, $\bm{R}_i = \rho(\bm{1}_{N_i}\bm{1}_{N_i}^\top) + (1-\rho)\textbf{I}_{N_i} \in \mathbb{R}^{N_i \times N_i}$ where $\rho$ is the ICC estimated by moment estimators $\hat{\rho}$. In the case of an independence working correlation, then $\rho=0$ and the resulting estimators will generally be referred to as independence estimating equations (IEE).

The CSLMM (Equation \ref{eq:ME_sat}) is a more specific case of the CSGEE (Equation \ref{eq:GEE_sat}) with an identity-link and constant working variance $v(Y_{ik})=\sigma^2$.
Broadly, these models with individual-level observations aim to target individual-level estimands and can be implemented with inverse cluster-size weights (CSLMMw, CSGEEw) to target cluster-level estimands.

\subsection{G-computation}
\label{sect:g-comp}

G-computation (also referred to as ``standardization'' or ``marginalization'') in its simplest form takes the average of fitted model predictions across all samples after setting the treatment assignment to either all-treated or all-control, then takes the average contrasts of those predictions. It has been commonly used as a population standardization technique to estimate marginal estimands \cite{rosenbaum_model-based_1987}.

In the general setting, cluster-size saturated GEE with g-computation treatment effect estimators and appropriate weighting ($\hat{\Delta}_{CSGEE-g}$, $\hat{\Delta}_{CSGEEw-g}$) are devised by g-computation and can be easily implemented in standard statistical software (e.g., using \texttt{predict()} in \texttt{R}; Supplementary Appendix \ref{app:example_code}).
As per Equation \ref{eq:GEE_sat} and with $a \in \{0,1\}$, the vector of all-control ($a=0$) and all-treated ($a=1$) g-computation predicted outcomes can then be specified as
\[
\hat{\bm{\mu}}_i(a) 
= g^{-1}\left( \bm{Q}_{i}(a) \hat{\bm{\beta}}  \right)
=  g^{-1}\left( \hat{\beta}_0 + a \hat{\beta}_Z + a N_i \hat{\beta}_{ZN} + N_i \hat{\beta}_{N}  + \bm{X}_{ik} \hat{\bm{\beta}}_{\bm{X}} \right)_{k:1,...,N_i}
\in \mathbb{R}^{N_i}
\,.
\]
The appropriately weighted estimators targeting a broad range of individual-level or cluster-level estimands can then be appropriately averaged and formulated with contrasts $h$ (Section \ref{sect:estimands}).

Finally, robust sandwich variance estimators can be specified with the delta method within an M-estimation framework \cite{ross_m-estimation_2024}. Alternatively, non-parametric robust variance estimators such as the ``leave-one-cluster-out'' jackknife can also be easily implemented \cite{bell_bias_2002} and have become a standard recommendation in such settings \cite{li_model-robust_2025}.

\subsection{Summary}

Broadly, this CS-g modeling approach consists of the appropriately weighted ``CSGEE$\bm{\lambda}$-g'' (CSGEE-g, CSGEEw-g) estimators, with ``CSLMM$\bm{\lambda}$-g'' (CSLMM-g, CSLMMw-g) estimators being specific incidences.
Altogether, the proposed CS-g modeling procedure is as follows
\begin{enumerate}
    \item Specify a cluster-size saturated linear mixed-effects model or GEE with appropriate weights (unweighted to target an individual-level estimand; inverse cluster-size weighted to target a cluster-level estimand).
    \item Implement g-computation with these appropriately weighted cluster-size saturated models to predict outcomes under all-treated or all-control conditions. Then contrast these predicted outcomes to estimate the marginal estimands.
    \item Implement a cluster-robust variance (e.g., sandwich variance estimator, non-parametric ``leave-one-cluster-out'' jackknife variance estimator, etc) for inference.
\end{enumerate}
An illustrative example of template code using \texttt{R} is included in the Supplementary Appendix (\ref{app:example_code}).

\section{Robust Consistency of CS-g}
\label{sect:proof}

We first outline Lemma \ref{lemma:variance}, a standard result which establishes the necessary regularity conditions for consistency of the cluster-robust sandwich variance estimator within an M-estimation framework \citep{van_der_vaart_asymptotic_1998, tsiatis_semiparametric_2006}.

\begin{Lemma}
\label{lemma:variance}

Let $\bm{O}_1,...,\bm{O}_{m}$ be i.i.d. samples from a common distribution on $O$.
Let $\bm{\psi}(\bm{O},\bm{\theta})$ be a known estimating equation with parameters $\bm{\theta} \in \Theta$, a compact set of Euclidean space.
Let $\hat{\bm{\theta}}$ be the solution to $\sum_{i=1}^{m} \bm{\psi}(\bm{O}_i, \bm{\theta})=0$.
We assume that $\bm{\psi}$ satisfies the following regularity conditions
\begin{enumerate}
    \item There exists a unique solution in the interior of $\Theta$, denoted as $\underline{\bm{\theta}}$, to the equation $E[\bm{\psi}(\bm{O},\bm{\theta})] = 0$.
    \item The function $\bm{\theta} \mapsto \bm{\psi}(o, \bm{\theta})$, together with its first and second derivatives, is dominated by a square-integrable function for every $o$ in the support of $\bm{O}$.
    \item $E\left[\frac{d\bm{\psi}(\bm{O},\bm{\theta})}{d\bm{\theta}^{\top}} \mid_{\bm{\theta}=\underline{\bm{\theta}}}\right]$ is invertible.
\end{enumerate}
Then we have
\[
\begin{split}
    \hat{\bm{\theta}} &\xrightarrow{P} \underline{\bm{\theta}} \\
    m^{1/2}(\hat{\bm{\theta}} - \underline{\bm{\theta}}) &\xrightarrow{d} N(0,\textbf{V})
\end{split}
\]
where $\textbf{V}=E[\text{IF}(\bm{O},\underline{\bm{\theta}})\text{IF}(\bm{O},\underline{\bm{\theta}})^\top]$ and $\text{IF}(\bm{O},\underline{\bm{\theta}}) = -\left( E\left[\frac{d\bm{\psi}(\bm{O},\bm{\theta})}{d\bm{\theta}^{\top}} \mid_{\bm{\theta}=\underline{\bm{\theta}}}\right]^{-1} \bm{\psi}(\bm{O},\underline{\bm{\theta}}) \right)$ is the influence function for $\hat{\bm{\theta}}$.

Furthermore, the sandwich variance estimator
\[
    \hat{\bm{V}} = m^{-1}\sum_{i=1}^{m} \widehat{\text{IF}}(\bm{O}_i,\hat{\bm{\theta}}) \widehat{\text{IF}}(\bm{O}_i,\hat{\bm{\theta}})^\top \xrightarrow{P} \textbf{V}
\]
where $\widehat{\text{IF}}(\bm{O}_i,\hat{\bm{\theta}})=-\left( 
\left[
    m^{-1}\sum_{i=1}^{m} \frac{d\bm{\psi}(\bm{O}_i,\bm{\theta})}{d\bm{\theta}^{\top}} \mid_{\bm{\theta}=\hat{\bm{\theta}}}
\right]^{-1}
\bm{\psi}(\bm{O}_i, \hat{\bm{\theta}})
\right)$.
\end{Lemma}

The complete proof for Lemma \ref{lemma:variance} is included in the Supplementary Appendix (\ref{app:proofs}).
While Lemma \ref{lemma:variance} demonstrates the consistency of the sandwich variance estimator, previous publications have highlighted that this variance estimator can be biased with small sample sizes and can yield under-coverage of the 95\% confidence intervals \cite{bell_bias_2002}. This is true when using either a normal approximation or $t$-distribution, even with Satterthwaite degrees of freedom adjustment, and has lead to development for the jackknife and bias-reduced linearization variance estimators for better coverage with a $t$-distribution \cite{bell_bias_2002}.
Throughout this manuscript, we will typically employ a non-parametric ``leave-one-cluster-out'' jackknife variance estimator \cite{li_model-robust_2025,bell_bias_2002}. Such a jackknife variance has similar asymptotic properties to the cluster-robust sandwich variance estimator \cite{bell_bias_2002} and its re-sampling procedure is simple to manually implement with standard statistical software.

With the consistency of the variance estimator established in Lemma \ref{lemma:variance}, we next need to establish consistency of the appropriately weighted CS-g point estimators for their corresponding estimands.

\begin{Theorem}
\label{Theorem_g}
    Under standard regularity conditions in Lemma \ref{lemma:variance} and assume Assumptions A1 and A2, then the following Central Limit Theorems hold for a P-CRT: (a) $\hat{V}_{CSLMM-g}^{-1/2} m^{1/2} (\hat{\Delta}_{CSLMM-g} - 
    \Delta_{iATE}
    ) \xrightarrow{d} N(0,1)$;  and (b) $\hat{V}_{CSLMMw-g}^{-1/2} m^{1/2} (\hat{\Delta}_{CSLMMw-g} - 
    \Delta_{cATE}
    ) \xrightarrow{d} N(0,1)$.
\end{Theorem}

Unlike the basic linear mixed-effects model (Equation \ref{eq:ME_basic}), the cluster-size saturated linear mixed-effects model (including treatment $\times$ cluster-size interactions and cluster-size main effects; Equation \ref{eq:ME_sat}) with g-computation and appropriate weighting (CSLMM-g: $\hat{\Delta}_{CSLMM-g}$, CSLMMw-g: $\hat{\Delta}_{CSLMMw-g}$) yields consistent estimators for the iATE and cATE ($\Delta_{iATE}, \Delta_{cATE}$) in a P-CRT with ICS by relying only on Assumptions A1 (Super-population sampling) and A2 (Cluster randomization).
Importantly, these consistency results are robust to arbitrary model-misspecification, including to the covariate and correlation structure, and do not require the functional form of the treatment $\times$ cluster-size interactions and cluster-size main effects to be correctly specified (unlike previous suggestions \cite{hooper_wood_2026}).
An in-depth proof of Theorem \ref{Theorem_g} is included in the Supplementary Appendix (\ref{app:proofs}). We briefly outline the proof below to highlight that it is specifically the inclusion of both the treatment $\times$ cluster-size interaction and cluster-size main effect  that produces the robust consistency result.

\noindent \textit{Proof of Theorem \ref{Theorem_g}.}
Denote $\bm{\theta}=(\Delta_{CSLMM\bm{\lambda}-g}, \bm{\beta}, \tau^2, \sigma^2)^{\top} \in \mathbb{R}^{7+p}$ as the vector of unknown parameters to be estimated by M-estimation.
The appropriately weighted estimator $\Delta_{CSLMM\bm{\lambda}-g}$ can be set with $\lambda_i=N_i$ or $\lambda_i=1$ to produce the unweighted (CSLMM-g) or inverse cluster-size weighted (CSLMMw-g) estimators, which we will demonstrate to be consistent for the iATE ($\Delta_{iATE}$) or cATE ($\Delta_{cATE}$) in P-CRTs with potential ICS.
The estimators are then the solution to the estimating equations $\sum_{i=1}^m \bm{\psi}(\bm{O}_i, \bm{\theta})$, where
\[
\begin{split}
    \bm{\psi}(\bm{O}_i, \bm{\theta})
    & = 
        \left(
        \begin{gathered}
            \left(\frac{\sum_{s=1}^{m}\lambda_s}{m}\right) \Delta_{CSLMM\bm{\lambda}-g} - \frac{\lambda_i}{N_i} \bm{1}_{N_i}^\top(\bm{\mu}_i(1) - \bm{\mu}_i(0)) \\
            \frac{\lambda_i}{N_i}\bm{Q}_i^\top \bm{V}_i (\bm{Y}_i - \bm{Q}_i\bm{\beta}) \\
            -tr(\bm{V}_i) + (\bm{Y}_i - \bm{Q}_i\bm{\beta})^\top \bm{V}^2_i (\bm{Y}_i - \bm{Q}_i\bm{\beta}) \\
            -\bm{1}_{N_i}^\top \bm{V}_i \bm{1}_{N_i} + (\bm{Y}_i - \bm{Q}_i\bm{\beta})^\top \bm{V}_i ( \bm{1}_{N_i} \bm{1}_{N_i}^\top) \bm{V}_i  (\bm{Y}_i - \bm{Q}_i\bm{\beta})
        \end{gathered}
        \right)
\end{split}
\]
\sloppy
with
$\bm{\mu}_i(a) = \bm{Q}_i(a) \bm{\beta} = \left(\beta_{0}  + a \beta_Z + a N_i  \beta_{ZN} + N_i \beta_{N} + \bm{\beta}_X^\top \bm{X}_{ik}\right)_{k=1,...,N_i} \in \mathbb{R}^{N_i}$, $a\in\{0,1\}$,
$\bm{V}_i = \bm{\Sigma}_i^{-1} \in \mathbb{R}^{N_i \times N_i}$, and $tr(\bm{V}_i)$ being the trace of $\bm{V}_i$.
With $\rho=\frac{\tau^2}{\tau^2 + \sigma^2}$ being the ICC, we can then easily demonstrate that $\bm{V}_i = \left(\frac{1}{\tau^2 + \sigma^2}\right)\left(\frac{1}{1-\rho}\right)\left(\textbf{I}_{N_i} - \bm{1}_{N_i} \bm{1}_{N_i}^\top (\rho/[1+(N_i-1)\rho])\right)$.
The maximum likelihood estimator for $\bm{\theta}$ is defined as a solution to the estimating equation
\[
    \sum_{i=1}^{m} \bm{\psi}(\bm{O}_i;\bm{\theta})=0 \,.
\]
For the estimating equation $\bm{\psi}$, we prove the convergence and asymptotic normality of $\hat{\bm{\theta}}$ by applying Lemma \ref{lemma:variance}, with its conditions for $\bm{\psi}$ assumed as regularity conditions.

We then denote $\underline{\bm{\theta}}=(\underline{\Delta}_{CSLMM\bm{\lambda}-g}, \underline{\bm{\beta}}_0, \underline{\beta}_Z, \underline{\beta}_{ZN}, \underline{\beta}_{N}, \underline{\bm{\beta}}_{X}^\top, \underline{\tau}^2, \underline{\sigma^2})^{\top}$ as the solution to $E[\bm{\psi}(\bm{O};\bm{\theta})]=0$, and prove $\underline{\Delta}_{CSLMM\bm{\lambda}-g} = \Delta_{\bm{\lambda}ATE}$ ($\underline{\Delta}_{CSLMM-g} = \Delta_{iATE}$, $\underline{\Delta}_{CSLMMw-g} = \Delta_{cATE}$) to imply robust consistency of the appropriately weighted CSLMM$\bm{\lambda}$-g estimator.
To proceed, the second through fifth entries of $E[\bm{\psi}(\bm{O};\bm{\theta})]=0$ are
\begin{gather}
\label{eq:baseline}
    E\left[\frac{\lambda}{N} \bm{1}_N^\top \underline{\bm{V}}(\bm{Y} - \bm{Q}\underline{\bm{\beta}})\right] = 0 \\
\label{eq:treatment}
    E\left[\frac{\lambda}{N} I\{Z=1\}\bm{1}_N^\top \underline{\bm{V}}(\bm{Y} - \bm{Q}\underline{\bm{\beta}})\right] = 0 \\
\label{eq:treatment_N}
    E\left[\lambda I\{Z=1\} \bm{1}_N^\top \underline{\bm{V}}(\bm{Y} - \bm{Q}\underline{\bm{\beta}})\right] = 0 \\
\label{eq:N}
    E\left[\lambda \bm{1}_N^\top \underline{\bm{V}}(\bm{Y} - \bm{Q}\underline{\bm{\beta}})\right] = 0
\end{gather}
corresponding to the score functions for the model intercept $\beta_0$ (Equation \ref{eq:baseline}), treatment effect coefficient $\beta_Z$ (Equation \ref{eq:treatment}), treatment $\times$ cluster-size interaction $\beta_{ZN}$ (Equation \ref{eq:treatment_N}), and cluster-size main effect $\beta_N$ (Equation \ref{eq:N}), respectively. $\underline{\bm{V}}$ is equal to $\bm{V}$ with $(\tau^2, \sigma^2, \rho)$ replaced by $(\underline{\tau}^2, \underline{\sigma}^2, 
\underline{\rho})$.
Recall that subscript $i$ is omitted when taking the expectation with respect to distribution $\mathcal{P}$ with Assumption A1.
With Assumption A2, we can then use Equations (\ref{eq:baseline}-\ref{eq:N}) to derive the following equalities:
\begin{equation}
\label{eq.ee.ZN}
    E\left[\lambda \left(\frac{1}{1+(N-1)\underline{\rho}}\right) \bm{1}_N^\top \left\{\left[\bm{Y}(1)-\bm{Y}(0)\right] - \left[\underline{\bm{\mu}}(1) - \underline{\bm{\mu}}(0)]\right]\right\}\right] = 0
\end{equation}
and
\begin{equation}
\label{eq.ee.Z}
    E\left[\frac{\lambda}{N} \left(\frac{1}{1+(N-1)\underline{\rho}}\right) \bm{1}_N^\top \left\{\left[\bm{Y}(1)-\bm{Y}(0)\right] - \left[\underline{\bm{\mu}}(1) - \underline{\bm{\mu}}(0)]\right]\right\}\right] = 0 \,.
\end{equation}

In contrast to the basic linear mixed-effects model (Equation \ref{eq:ME_basic}) and models including only a cluster-size main effect, it is the presence of both the treatment $\times$ cluster-size interaction $\beta_{ZN}$ (Equation \ref{eq:treatment_N}) and cluster-size main effect $\beta_N$ (Equation \ref{eq:N})  in the CSLMM$\bm{\lambda}$-g estimator that produces Equation (\ref{eq.ee.ZN}), which alongside Equation (\ref{eq.ee.Z}), yields the final result:
\[
\begin{split}
   \underline{\Delta}_{CSLMM\bm{\lambda}-g} &= \frac{E[(\lambda/N)\bm{1}_N^\top \{\underline{\bm{\mu}}(1) - \underline{\bm{\mu}}(0)\}]}{E[\lambda]} \\
    &= \frac{E[(\lambda/N)\bm{1}_N^\top \{\bm{Y}(1) - \bm{Y}(0)\}]}{E[\lambda]} = \Delta_{\bm{\lambda}ATE} 
\end{split}
\]
completing the proof of consistency. $\square$

\subsection{Generalized consistency of CS-g}
\label{sect:iATE_cATE_proof-g}

We further establish robust inference for the iATE and cATE ($\Delta_{iATE}, \Delta_{cATE}$) via the appropriately weighted cluster-size saturated generalized estimating equation with g-computation (CSGEE$\bm{\lambda}$-g: CSGEE-g, CSGEEw-g), as described in (Equation \ref{eq:GEE_sat}).

\begin{Lemma}
\label{lemma:g}
    Assuming either (I) working independence $\rho=0$, (II) $g$ is the identity link with constant working variance $v(Y_{ik})=\sigma^2$, or (III) $\bm{X}_{ik}$ only contains cluster-level information,
    then $\bm{\mathcal{Z}}_i^{-1/2} \bm{R}_i^{-1} \bm{\mathcal{Z}}_i^{-1/2} = \bm{\mathcal{Z}}_i^{-1}\bm{R}_i^{-1}$
\end{Lemma}

Lemma \ref{lemma:g} has been similarly invoked in previous work \cite{wang_how_2024,lee_fixed-effects_2026} and allows the proof for Theorem \ref{Theorem_GEE_g} to proceed in a similar manner to the proof for Theorem \ref{Theorem_g}.
The complete proofs for Lemma \ref{lemma:g} and Theorem \ref{Theorem_GEE_g} are included in the Supplementary Appendix (\ref{app:proofs}).

\begin{Theorem}
\label{Theorem_GEE_g}
    Under standard regularity conditions in Lemma \ref{lemma:variance}, assume Assumptions A1 and A2 alongside supplementary assumptions specified in Lemma \ref{lemma:g}: (I) working independence $\rho=0$, (II) $g$ is the identity link with constant working variance $v(Y_{ik})=\sigma^2$, or (III) $\bm{X}_{ik}$ only contains cluster-level information,
    or alternatively (IV) the mean model (\ref{eq:GEE_sat}) is correctly specified.
    Then the following Central Limit Theorems hold for a P-CRT: (a) $\hat{V}_{CSGEE-g}^{-1/2} m^{1/2} (\hat{\Delta}_{CSGEE-g} - 
    \Delta_{iATE}
    ) \xrightarrow{d} N(0,1)$;  and (b) $\hat{V}_{CSGEEw-g}^{-1/2} m^{1/2} (\hat{\Delta}_{CSGEEw-g} - 
    \Delta_{cATE}
    ) \xrightarrow{d} N(0,1)$.
\end{Theorem}

Unlike the basic GEE with an exchangeable working correlation, a cluster-size saturated GEE (with treatment $\times$ cluster-size interactions and cluster-size main effects; Equation \ref{eq:GEE_sat}) with g-computation, and appropriate weighting ($\hat{\Delta}_{CSGEE-g}$, $\hat{\Delta}_{CSGEEw-g}$) yield consistent estimators for the iATE and cATE ($\Delta_{iATE}, \Delta_{cATE}$) in a P-CRT with ICS by relying only on Assumptions A1 (Super-population sampling), A2 (Cluster randomization),
and supplementary assumptions (II) or (III).
Theorem \ref{Theorem_g} can be interpreted as a specific instance of Theorem \ref{Theorem_GEE_g} under supplementary assumption (II).
Under supplementary assumptions (I)-(III), these consistency results do not require the functional form of the treatment $\times$ cluster-size interactions and cluster-size main effects to be correctly specified, and are robust to other arbitrary model-misspecification, including to the covariate and correlation structure.

As per Wang et al. \cite{wang_model-robust_2024} (and extensions \cite{wang_how_2024,lee_fixed-effects_2026}), the consistency results described here generally extend to alternative effect measures, including marginal risk ratio and odds ratio estimands.
This is easily achieved by simply changing the contrast $h$ of the g-computation elements to match that of the target estimand (Section \ref{sect:estimands}).

\section{Exact finite-sample equivalency to model-robust standardization}
\label{sect:g=MRS}

Notably, we can prove that the described consistent CS-g estimators (Theorems \ref{Theorem_g} \& \ref{Theorem_GEE_g}) are equivalent to the corresponding estimators from cluster-size saturated models with model-robust standardization (CS-MRS).
Broadly, the appropriately weighted CSGEE$\bm{\lambda}$-g estimator (Section \ref{sect:CS-g}) can be re-written as
\begin{equation}
\label{eq:csg}
\widehat\Delta_{CSGEE\bm{\lambda}-g}
=h\left\{ \sum_{i=1}^{m}\frac{\lambda_i}{\sum_{s=1}^m\lambda_s} \widehat{\bar\mu}_i(1), \sum_{i=1}^{m}\frac{\lambda_i}{\sum_{s=1}^m\lambda_s} \widehat{\bar\mu}_i(0) \right\} \,,
\end{equation}
with contrasts $h$ (Section \ref{sect:estimands}), $\widehat{\bar\mu}_i(a)=N_i^{-1}\bm{1}_{N_i}^\top\hat{\bm{\mu}}_i(a)$ being the cluster average of the g-computation predictions under intervention level $a\in\{0,1\}$,
and $\hat{\bm{\mu}}_i(a)=g^{-1}\left(\bm{Q}_i(a)\widehat{\bm{\beta}}\right)$ where all units in cluster $i$ are set to $Z_i=a$ (Section \ref{sect:g-comp}).
Recall that this returns $\widehat\Delta_{CSGEE-g}$ when $\lambda_i=N_i$, and $\widehat\Delta_{CSGEEw-g}$ when $\lambda_i=1$.

The appropriately weighted cluster-size saturated GEE can alternatively be implemented within the model-robust standardization described by Li et al. \cite{li_model-robust_2025} (CSGEE$\bm{\lambda}$-MRS: CSGEE-MRS, CSGEEw-MRS).
The appropriately weighted CSGEE$\bm{\lambda}$-MRS estimator can then be broadly specified as
\begin{equation}
\label{eq:CSGEE-MRS}
\widehat\Delta_{CSGEE\bm{\lambda}-MRS}
= h\left\{ \widehat{\mu}_{CSGEE\bm{\lambda}-MRS}(1),  \widehat{\mu}_{CSGEE\bm{\lambda}-MRS}(0) \right\}
\end{equation}
where
\begin{equation}
\label{eq:CSGEE-MRS_A}
\widehat{\mu}_{CSGEE\bm{\lambda}-MRS}(a) = \sum_{i=1}^{m} \frac{\lambda_i}{\sum_{s=1}^m\lambda_s}
\Biggl\{
\widehat{\bar\mu}_i(a)
+\underbrace{\frac{I\{Z_i=a\}\bigl(\bar Y_i-\widehat{\bar\mu}_i(a)\bigr)}
{\pi^{a}(1-\pi)^{1-a}}}_{\text{weighted cluster-level residual}}
\Biggr\} \,.
\end{equation}
Here, $\bar Y_i=N_i^{-1}\bm{1}_{N_i}^\top\bm{Y}_i$ is the observed cluster average outcome incorporated into the additional weighted cluster-level residual term.
We specify constant cluster randomization probability $\pi\in(0,1)$ (Assumption A2; e.g., $\pi=1/2$ under 1:1 cluster randomization).

\begin{Theorem}
\label{Theorem_g=MRS}
    Under standard regularity conditions in Lemma \ref{lemma:variance}, assume Assumptions A1 and A2 alongside supplementary assumptions specified in Lemma \ref{lemma:g}: (I) working independence $\rho=0$, (II) $g$ is the identity link with constant working variance $v(Y_{ik})=\sigma^2$, or (III) $\bm{X}_{ik}$ only contains cluster-level information,
    or alternatively (IV) the mean model (\ref{eq:GEE_sat}) is correctly specified.
    Then the following estimators have exact finite-sample equivalence: $\hat{\Delta}_{CSGEE-g} = \hat{\Delta}_{CSGEE-MRS}$ and $\hat{\Delta}_{CSGEEw-g} = \hat{\Delta}_{CSGEEw-MRS}$.
\end{Theorem}

The complete proofs for Theorem \ref{Theorem_g=MRS} are included in the Supplementary Appendix (\ref{app:proofs}).
Unlike the basic GEE, the consistent cluster-size saturated GEE yields equivalent g-computation and MRS estimators by canceling out the additional MRS weighted cluster-level residual.
Usefully, the proposed consistent CS-g estimators and their jackknife variances can then be easily implemented within the previously developed \texttt{MRStdCRT} package in \texttt{R} \cite{li_model-robust_2025}.

\section{Simulation}
\label{sect:sim}

We describe some simulation studies to empirically illustrate our theoretical findings.
We simulate 1000 replicates of P-CRTs for each of the following simulation scenarios.
Simulation scenarios 1 and 2 generate continuous and binary outcomes, respectively, in P-CRTs with ICS and $m=40$ clusters.
Separately, we also replicate some simulation scenarios described in Li et al. \cite{li_model-robust_2025} with ICS, $m=30$ clusters, and continuous or binary outcomes.
Full simulation results are included in the Supplementary Appendix (\ref{app:extended_sim}).

Simulation scenarios with continuous outcomes are analyzed with the appropriately weighted 
linear independence estimating equations (IEE$\bm{\lambda}$: IEE, IEEw; equivalent here to a linear models using ordinary least squares), 
linear mixed-effects models (LMM$\bm{\lambda}$: LMM, LMMw; Equation \ref{eq:ME_basic}), 
cluster-size main effect linear mixed-effects models (CMLMM$\bm{\lambda}$: CMLMM, CMLMMw; Equation \ref{eq:ME_basic} + a continuous cluster-size main effect),
and
linear mixed-effects models with model-robust standardization (LMM$\bm{\lambda}$-MRS: LMM-MRS, LMMw-MRS) \cite{li_model-robust_2025}.
We compare these results against the appropriately weighted 
cluster-size saturated linear mixed-effects model with g-computation (CSLMM$\bm{\lambda}$-g: CSLMM-g, CSLMMw-g).

Simulation scenarios with binary outcomes are analyzed with the appropriately weighted logit-link IEE$\bm{\lambda}$, generalized linear mixed-effects models with a cluster random effect (GLMM$\bm{\lambda}$: GLMM, GLMMw), GEEs (GEE$\bm{\lambda}$: GEE, GEEw; with an exchangeable working correlation structure),
cluster-size main effect GEEs (CMGEE$\bm{\lambda}$: CMGEE, CMGEEw),
and GEEs with model-robust standardization (GEE$\bm{\lambda}$-MRS: GEE-MRS, GEEw-MRS) \cite{li_model-robust_2025}.
We compare these results against the appropriately weighted cluster-size saturated GEEs with g-computation (CSGEE$\bm{\lambda}$-g: CSGEE-g, CSGEEw-g) and cluster-size saturated generalized linear mixed-effects models with g-computation (CSGLMM$\bm{\lambda}$-g: CSGLMM-g, CSGLMMw-g).

With these estimators, we report the percent relative bias, empirical variance, average ``leave-one-cluster-out'' jackknife variance estimates, and determine the coverage probability of the 95\% confidence intervals with the $m-2$ degrees of freedom $t$-distribution \cite{ford_maintaining_2020}.
The jackknife variance estimator is implemented as described in Bell \& McCaffrey \cite{bell_bias_2002}.
Specifically, the jackknife variance estimator benefits from easily generalizable implementation across different statistical software.
Finally, we also run a test for informative cluster size between the unweighted and inverse cluster-size weighted estimators, as outlined in Li et al. \cite{li_model-robust_2025}.
This ICS test statistic is the ratio of the difference between the unweighted and inverse cluster-size weighted proposed estimator over the jackknife standard error of the difference, and approximately takes on a $t$-distribution which we specify with $m-2$ degrees of freedom.

\subsection{Simulation Scenario 1}
\label{sect:sim_scenario_1}

In simulation scenario 1, we sample $m=40$ clusters with clusters emerging from subpopulations $u \in \{A,B,C,D\}$, where $P(u)=0.25 \, \forall \, u\in\{A,B,C,D\}$.
Cluster sizes $N_{i,u}$, corresponding to subpopulation $u$, are then $N_{i,A} \sim Poisson(20)$, $N_{i,B} \sim Poisson(60)$, $N_{i,C} \sim Poisson(100)$, and $N_{i,D} \sim Poisson(120)$.
The P-CRT data is then generated with random 1:1 cluster assignment to treatment, ICS.I, and ICS.II, as
\[
Y_{ik,u} = \beta_{0,u} + I\{Z_i=1\}  \beta_{Z,u} + I\{X_{1,u} = 1\}\beta_{X_1,u} + sin(X_{2,u})\beta_{X_2,u} + \alpha_{i,u} + \epsilon_{ik} \,,
\]
where subpopulation-specific intercepts are $\beta_{0,u} \in \{\beta_{0,A}, \beta_{0,B},\beta_{0,C},\beta_{0,D}\} = \{2, 3, 4,5\}$.
Subpopulation-specific treatment effects are $\beta_{Z,u} \in \{\beta_{Z,A}, \beta_{Z,B},\beta_{Z,C},\beta_{Z,D}\} = \{0.5, 1,5,3\}$,
which targets $iATE = (0.5*20 + 1*60 + 5*100 + 3*120)/(20+60+100+120) = 3.1$ and $cATE = (0.5 + 1 + 5 + 3)/4 = 2.375$.
Subpopulation covariates are generated as 
$X_{1,u} \sim Bernoulli(p_{1,u})$, where $p_{1,u} \in \{p_{1,A}, p_{1,B}, p_{1,C}, p_{1,D}\} = \{0.2, 0.4, 0.6, 0.8\}$, $\beta_{X_1,u} \in \{\beta_{X_1,A}, \beta_{X_1,B}, \beta_{X_1,C}, \beta_{X_1,D}\} = \{0.2, 0.4, 0.5, 0.3\}$,
and
$X_{2,u} \sim Poisson(p_{2,u})$, where $p_{2,u} \in \{p_{2,A}, p_{2,B}, p_{2,C}, p_{2,D}\} = \{15, 50, 100, 200\}$, $\beta_{X_2,u} \in \{\beta_{X_2,A}, \beta_{X_2,B}, \beta_{X_2,C},\beta_{X_2,D}\} = \{1.2, 0.8, 0.6, 0.4\}$.
Finally, cluster random intercepts and residuals are generated as $\alpha_{i} \sim N(0,\tau^2)$ and $\epsilon_{ik} \sim N(0,\sigma^2)$, given ICC $\rho=0.05$, residual variance $\sigma^2=1$, and cluster variance $\tau^2=\rho\sigma^2/(1-\rho)$.

\subsection{Simulation Scenario 2}
\label{sect:sim_scenario_2}

In simulation scenario 2, we sample $m=40$ clusters where cluster subpopulations $u$ and cluster sizes $N_{i,u}$ are specified as in Scenario 1.
The P-CRT data is then generated with random 1:1 cluster assignment to treatment, ICS.I, and ICS.II, as
\[
logit(E[Y_{ik,u}|\alpha_{i,u}]) = \beta_{0,u} + I\{Z_i=1\}  \beta_{Z,u} + I\{X_{1,u} = 1\}\beta_{X_1,u} + \alpha_{i,u} \,,
\]
where subpopulation-specific intercepts are $\beta_{0,u} \in \{\beta_{0,A}, \beta_{0,B},\beta_{0,C},\beta_{0,D}\} = \{0.2, 0.3, 0.4, 0.5\}$.
Subpopulation-specific treatment effects are $\beta_{Z,u} \in\{\beta_{Z,A}, \beta_{Z,B},\beta_{Z,C},\beta_{Z,D}\} = \{0.2, 0.8 , 2, 1.2\}$.
Subpopulation covariates are generated as 
$X_{1,u} \sim Bernoulli(p_{1,u})$, where $p_{1,u} \in \{p_{1,A}, p_{1,B}, p_{1,C}, p_{1,D}\} = \{0.2, 0.4, 0.6, 0.8\}$, $\beta_{X_1,u} \in \{\beta_{X_1,A}, \beta_{X_1,B}, \beta_{X_1,C}, \beta_{X_1,D}\} = \{0.2, 0.4, 0.5, 0.3\}$.
Finally, cluster random intercepts are generated as $\alpha_{i} \sim N(0,\tau^2 = 0.5)$.
Altogether, this data-generating process targets the following marginal odds ratio estimands of $iAOR = 3.18$ and $cAOR = 2.33$ (Section \ref{sect:estimands}).

\subsection{Additional Simulation Scenarios}
\label{sect:sim_scenario_Li}

We additionally replicate simulation scenarios as described in Li et al. \cite{li_model-robust_2025} with ICS, $m=30$ clusters, and continuous or binary outcomes (Supplementary Appendix \ref{app:extended_sim}).
These simulation scenarios specified complex cluster-level and individual-level covariates relationships in the underlying data-generating processes.
In the replicate of Li et al.'s binary outcome simulation \cite{li_model-robust_2025}, we additionally fit CSGEE$\bm{\lambda}$-g with adjustment for only cluster-level covariates and all covariates (including individual level covariates) to empirically test violations of supplementary assumption (III) (Theorem \ref{Theorem_GEE_g}).

\subsection{Simulation results}

Figure \ref{fig:scenario_SWCRT_cont} includes the results from P-CRT simulation scenario 1 with continuous outcomes, ICS.I, and ICS.II (Section \ref{sect:sim_scenario_1}).
As expected \cite{wang_two_2022}, the appropriately weighted LMM$\bm{\lambda}$ and CMLMM$\bm{\lambda}$ produced biased results for the iATE and cATE.
Meanwhile, the appropriately weighted IEE$\bm{\lambda}$ \cite{wang_two_2022} and LMM$\bm{\lambda}$-MRS \cite{li_model-robust_2025} produced unbiased results for the iATE and cATE.
The proposed appropriately weighted CSLMM$\bm{\lambda}$-g produced unbiased results that were considerably more efficient than the previously validated IEE$\bm{\lambda}$ and LMM$\bm{\lambda}$-MRS approaches.

\begin{figure}[h]
    \centering
    \includegraphics[width=0.8\linewidth]{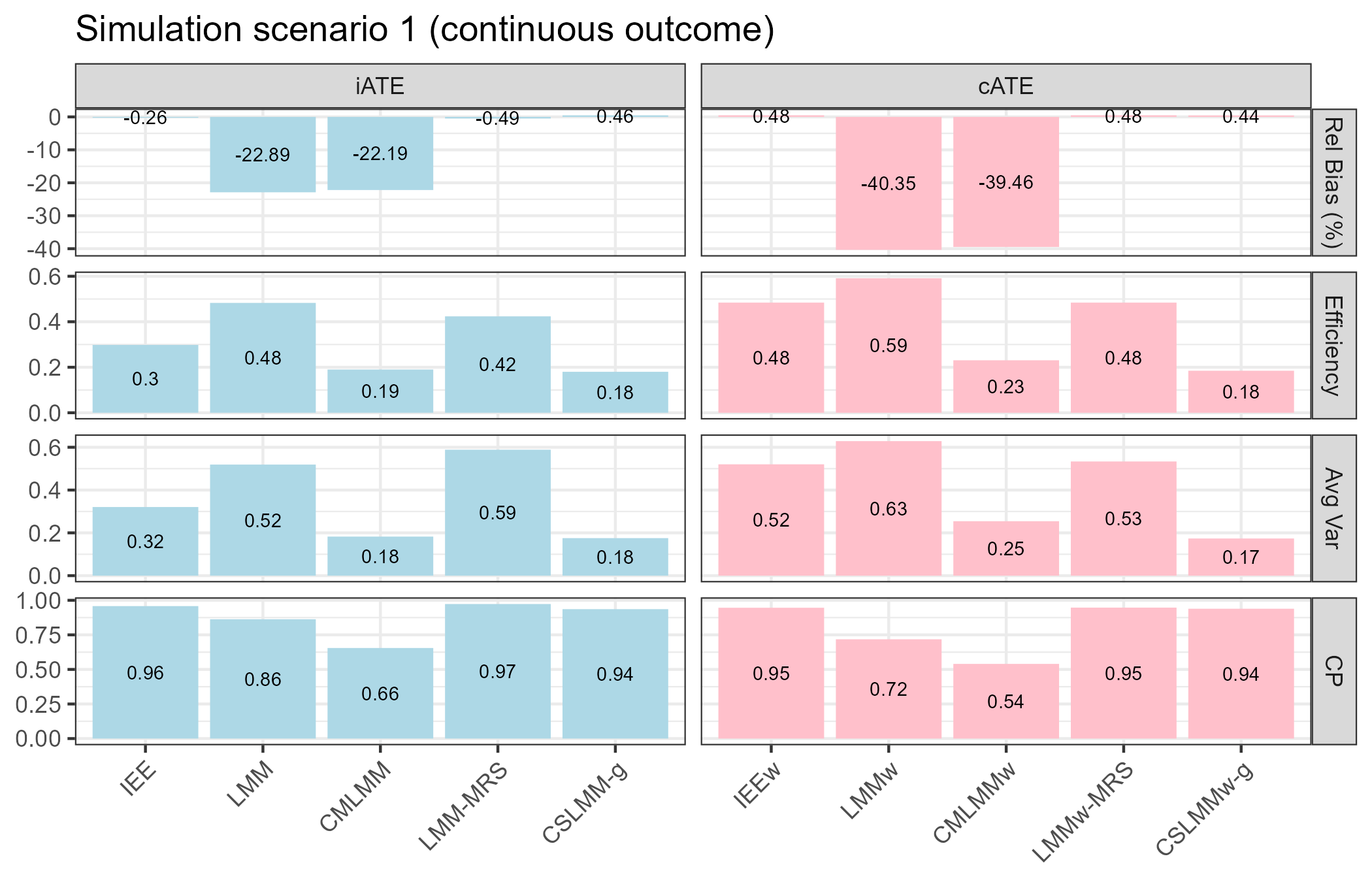}
    \caption{
        Analysis results from simulation scenario 1 of a P-CRT with continuous outcomes, targeting the individual and cluster-average treatment effects (iATE, cATE). 
        Results are reported for (i.) relative bias (\%), (ii.) empirical variance (``Efficiency''), (iii.) average jackknife variance estimates (``Avg Var''), and (iv.) coverage probabilities using the jackknife variance with $t(m-2)$
        (``CP'').
    }
    \label{fig:scenario_SWCRT_cont}
\end{figure}

Figure \ref{fig:scenario_SWCRT_bin} includes the results from P-CRT simulation scenario 2 with binary outcomes, ICS.I, and ICS.II (Section \ref{sect:sim_scenario_2}).
As expected \cite{wang_two_2022}, the appropriately weighted GLMM$\bm{\lambda}$, GEE$\bm{\lambda}$, and CMGEE$\bm{\lambda}$ produced biased results for their correspondingly weighted estimands (iAOR, cAOR). Meanwhile, the appropriately weighted IEE$\bm{\lambda}$ \cite{wang_two_2022} and GEE$\bm{\lambda}$-MRS \cite{li_model-robust_2025} produced unbiased results for their correspondingly weighted estimands.
The proposed appropriately weighted CSGEE$\bm{\lambda}$-g, alongside the CSGLMM$\bm{\lambda}$-g, produced unbiased results that were considerably more efficient than the previously validated IEE$\bm{\lambda}$ and GEE$\bm{\lambda}$-MRS approaches.

\begin{figure}[h]
    \centering
    \includegraphics[width=0.8\linewidth]{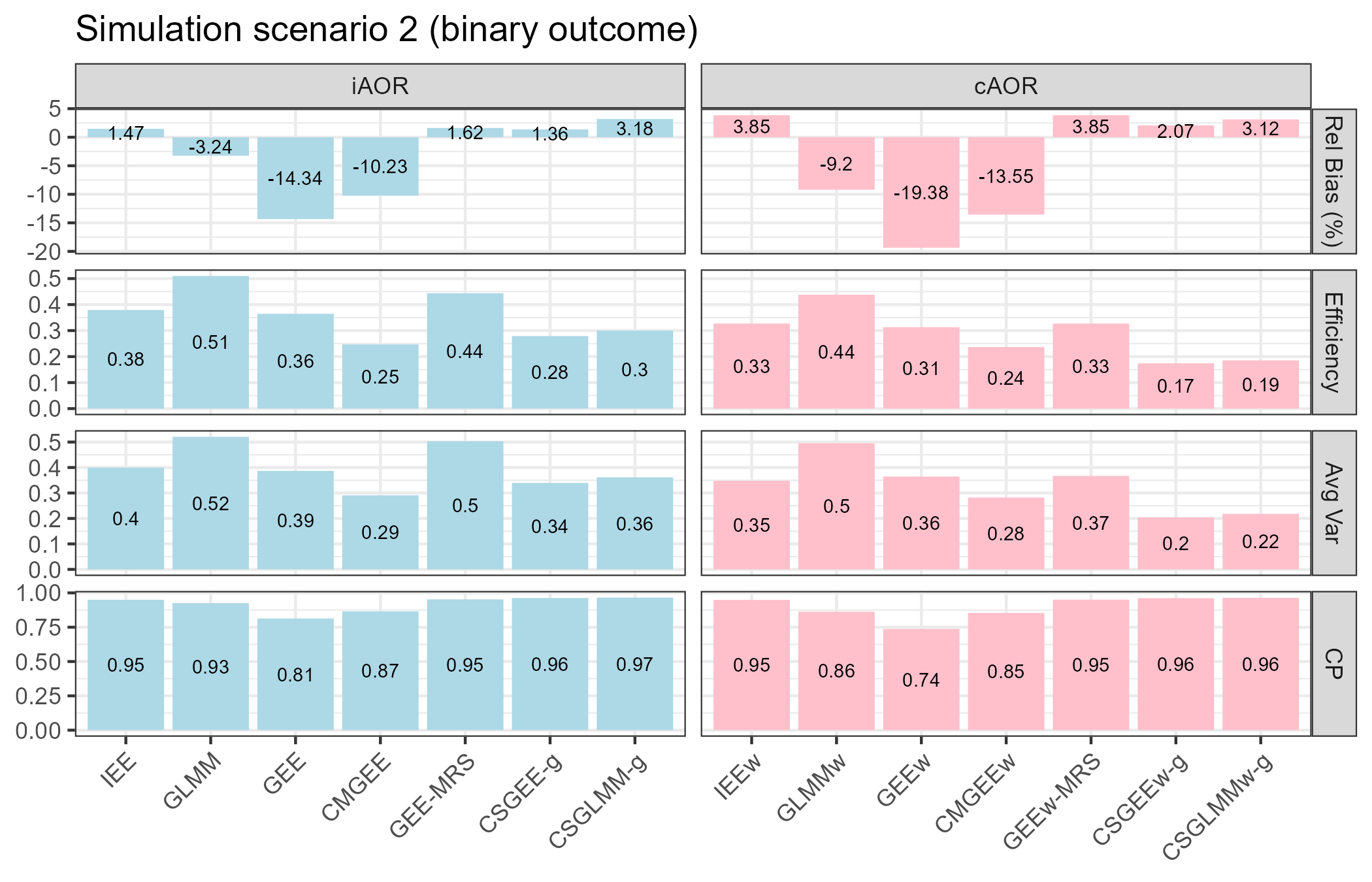}
    \caption{
        Analysis results from simulation scenario 2 of a P-CRT with binary outcomes, targeting the individual and cluster-average odds ratios (iAOR, cAOR). 
        Results are reported for (i.) relative bias (\%), (ii.) empirical variance (``Efficiency''), (iii.) average jackknife variance estimates (``Avg Var''), and (iv.) coverage probabilities using the jackknife variance with $t(m-2)$
        (``CP'').
    }
    \label{fig:scenario_SWCRT_bin}
\end{figure}

The leave-one-cluster-out jackknife variance generally approximated the empirical variance and yielded nominal coverage probabilities for unbiased estimators (Figures \ref{fig:scenario_SWCRT_cont} \& \ref{fig:scenario_SWCRT_bin}).
Unexpectedly, in 2 out of the 1000 simulation replicates, the CSGEEw-g in simulation scenario 2 was observed to yield jackknife variances that failed to converge (and $\rightarrow \infty$), with those 2 replicates being excluded from the reported results (Figure \ref{fig:scenario_SWCRT_bin}).

Overall, the results of our simulations with ICS.I and ICS.II validate our theoretical derivations (Section \ref{sect:proof}) and demonstrate the unbiasedness of the appropriately weighted CSLMM$\bm{\lambda}$-g for the iATE and cATE, and CSGEE$\bm{\lambda}$-g for the iAOR and cAOR.
The simulation results also empirically demonstrate minimal bias with the appropriately weighted CSGLMM$\bm{\lambda}$-g for the iAOR and cAOR; although the corresponding CSGEE$\bm{\lambda}$-g tended to have better performance in terms of bias and efficiency.
In contrast to the bias exhibited by CMLMM$\bm{\lambda}$ and CMGEE$\bm{\lambda}$, we empirically highlight that the unbiasedness of CSLMM$\bm{\lambda}$-g and CSGEE$\bm{\lambda}$-g requires both the treatment $\times$ cluster-size interaction term and cluster-size main effect in the specified model (Theorems \ref{Theorem_g} \& \ref{Theorem_GEE_g}).
Notably, these CS-g estimators were observed to be more efficient than the existing consistent approaches: IEE$\bm{\lambda}$ \cite{wang_two_2022}, LMM$\bm{\lambda}$-MRS, and GEE$\bm{\lambda}$-MRS \cite{li_model-robust_2025}.
As a result, the CS-g estimators also yielded higher average test-statistics and were more powerful for detecting ICS across simulation scenarios 1 and 2 (Table \ref{tab:ICS_test}).

\begin{table}[H]
\caption{
    Power and average values for the test of informative cluster size between the unweighted and inverse cluster-size weighted estimators.
    Results are reported for simulation scenarios 1 \& 2 with $m=40$ clusters across 1000 simulation replicates.
}
\label{tab:ICS_test}
\begin{center}
\bgroup
\def\arraystretch{1.3}
{
\begin{tabular}{|c c c|} 
    \hline
    \textbf{Estimator} & \textbf{Average Test Statistic} & \textbf{Power} \\
    \hline\hline
    \multicolumn{3}{|c|}{Scenario 1 (continuous)} \\
    \hline\hline
    IEE vs IEEw & -1.981 & 0.483\\
    \hdashline
    LMM-MRS vs LMMw-MRS & -1.592 & 0.245\\
    \hdashline
    CSLMM-g vs CSLMMw-g & -3.572 & 0.985\\
    \hline\hline
    \multicolumn{3}{|c|}{Scenario 2 (binary)} \\
    \hline\hline 
    IEE vs IEEw & -2.200 & 0.619\\
    \hdashline
    GEE-MRS vs GEEw-MRS & -1.790 & 0.375 \\
    \hdashline
    CSGLMM-g vs CSGLMMw-g & -2.444 & 0.727 \\
    \hdashline
    CSGEE-g vs CSGEEw-g & -2.433 & 0.728 \\
    \hline
\end{tabular}
}
\egroup
\end{center}
\end{table}

In the replication of Li et al.'s \cite{li_model-robust_2025} continuous outcome simulation, the unbiased results, efficiency gains, and greater power for detecting ICS with the appropriately weighted CSLMM$\bm{\lambda}$-g over existing consistent approaches are similarly observed (Supplementary Appendix \ref{app:extended_sim}).
Likewise, in the replication of Li et al.'s \cite{li_model-robust_2025} binary outcome simulation, the unbiased results, efficiency gains, and greater power for detecting ICS with the appropriately weighted CSGEE$\bm{\lambda}$-g and CSGLMM$\bm{\lambda}$-g over existing consistent approaches are also observed (Supplementary Appendix \ref{app:extended_sim}).
Additional adjustment for only cluster-level covariates in CSGEE$\bm{\lambda}$-g (maintaining Theorem \ref{Theorem_GEE_g}) yielded unbiased results with notable improvements in efficiency despite misspecification of the cluster-level covariate structure (Supplementary Appendix \ref{app:extended_sim}).
Surprisingly, the CSGEE$\bm{\lambda}$-g with adjustment for both cluster and individual-level covariates can still return minimally biased results in this specific binary data-generating process, despite violating all supplementary assumptions in Theorem \ref{Theorem_GEE_g} (Supplementary Appendix \ref{app:extended_sim}).
More discussion characterizing the bias when including individual-level covariates in CSGEE$\bm{\lambda}$-g is included in the Supplementary Appendix (\ref{app:extended_sim}).

Altogether, these simulation results support the consistent CS-g estimators as simpler, unbiased, and more efficient alternatives to existing approaches for targeting interpretable individual and cluster-level estimands in P-CRTs with ICS.

\section{Case study re-analysis}
\label{sect:case}

As in Li et al. \cite{li_model-robust_2025}, we re-analyzed the Pain Program of Active Coping and Training (PPACT) P-CRT \cite{debar_primary_2022}.
This data was collected from a mixed-methods, pragmatic P-CRT conducted at Kaiser Permanente, examining an interdisciplinary cognitive behavioral therapy approach for chronic pain patients in contrast to the usual care of long-term opioid treatment.
105 primary care providers served as the clusters, with 1:1 cluster randomization to either the PPACT intervention or usual care.
Data was collected from over 700 patients with cluster sizes ranging from 2 to 12 individuals.
Patients were followed for 12 months, with measurements taken every 3 months. We only focused on the final measurements at 12 months.
The target outcome for this re-analysis was the PEGS (Pain, Enjoyment, General Activity) score, a composite of pain intensity and interference, with values ranging from 0 to 40 which we treated as continuous.

We define the target estimand as the iATE and cATE on the PEGS score at 12 months.
We targeted these estimands with appropriately weighted independence estimating equations (IEE, IEEw), linear mixed-effects models (LMM, LMMw), linear mixed-effects models with model-robust standardization (LMM-MRS, LMMw-MRS), and cluster-size saturated linear mixed-effects models with g-computation (CSLMM-g, CSLMMw-g).
The point estimator, jackknife variance, and corresponding 95\% confidence intervals are included in Table \ref{tab:case_study_estimators}.
Furthermore, the test statistics for detecting ICS and corresponding p-values for each of the described estimators are included in Table \ref{tab:case_study_comparisons}.

\begin{table}[H]
\centering
\caption{Results from a re-analysis of the PPACT P-CRT are reported in terms of the point estimates, jackknife variance, and 95\% confidence intervals.}
\label{tab:case_study_estimators}
\begin{tabular}{llrrr}
\toprule
 & Estimator & Estimate & Var & 95\% CI \\
\midrule
 & IEE       & -0.633 & 0.035 & (-1.004, -0.262) \\
 & LMM       & -0.651 & 0.036 & (-1.026, -0.276) \\
 & LMM-MRS   & -0.632 & 0.035 & (-1.005, -0.260) \\
 & CSLMM-g   & -0.647 & 0.035 & (-1.020, -0.273) \\
\midrule
 & IEEw      & -0.702 & 0.041 & (-1.103, -0.301) \\
 & LMMw      & -0.726 & 0.046 & (-1.150, -0.301) \\
 & LMMw-MRS  & -0.702 & 0.041 & (-1.105, -0.299) \\
 & CSLMMw-g  & -0.728 & 0.041 & (-1.128, -0.328) \\
\bottomrule
\end{tabular}
\end{table}

\begin{table}[H]
\centering
\caption{The values of the test-statistic for detecting ICS and p-values with the different modeling approaches from a re-analysis of the PPACT P-CRT.}
\label{tab:case_study_comparisons}
\begin{tabular}{lrr}
\toprule
Comparison & Test statistic & p-value \\
\midrule
IEEw vs.\ IEE             & -0.929 & 0.355 \\
LMMw vs.\ LMM             & -0.889 & 0.376 \\
LMMw-MRS vs.\ LMM-MRS     & -0.922 & 0.359 \\
CSLMMw-g vs.\ CSLMM-g & -1.060 & 0.291 \\
\bottomrule
\end{tabular}
\end{table}

While the unweighted and weighted analyses appeared to have slight differences (Table \ref{tab:case_study_estimators}) our re-analyses failed to detect ICS in the PPACT trial (Table \ref{tab:case_study_comparisons}).
Our proposed CSLMM-g and CSLMMw-g approaches largely corresponded with the other analysis results and were similarly efficient.
Still, the test statistic for detecting ICS was larger in magnitude and produced smaller p-values when using the CSLMM-g and CSLMMw-g models, compared to the other consistent approaches. This corresponds with the simulation results that we reported earlier.
Extended re-analysis results are included in the Supplementary Appendix \ref{app:extended_case}.

\section{Discussion}
\label{sect:discussion}

In this work, we propose ``cluster-size saturated models with g-computation'' (CS-g) as a simple parametric strategy that allows the familiar linear mixed-effects model and GEE with an exchangeable working correlation to consistently target the individual-average and cluster-average treatment effects (iATE, cATE) in parallel cluster randomized trials (P-CRTs), while explicitly allowing for informative cluster sizes (ICS).
The approach requires only two simple modifications to standard practice: (1.) augmenting the working model with a cluster-size $\times$ treatment interaction and cluster-size main effect, and (2.) applying g-computation to recover an interpretable marginal estimand.

We rigorously establish the model-robustness of the appropriately weighted cluster-size saturated linear mixed-effects model with g-computation (CSLMM$\bm{\lambda}$-g: CSLMM-g, CSLMMw-g) and the more general cluster-size saturated GEE with g-computation (CSGEE$\bm{\lambda}$-g: CSGEE-g, CSGEEw-g) to arbitrary misspecification of the correlation and covariate structure, including misspecification of the functional form of the cluster-size $\times$ treatment interaction and cluster-size main effect (Theorems \ref{Theorem_g} \& \ref{Theorem_GEE_g}).
These consistency results rely only on a super-population sampling assumption (A1) and cluster randomization (A2), and importantly do not invoke a non-informative cluster size assumption.
Specifically, it is the inclusion of both the cluster-size $\times$ treatment interaction and cluster-size main effect in the cluster-size saturated models that yields this model-robust consistency, which we prove for Theorems \ref{Theorem_g} \& \ref{Theorem_GEE_g} and empirically observe in the simulation study results.
We also broadly demonstrate that under the outlined conditions, the appropriately weighted CSGEE$\bm{\lambda}$-g estimators have exact finite-sample equivalence to the appropriately weighted CSGEE$\bm{\lambda}$-MRS estimators (Theorem \ref{Theorem_g=MRS}).
Our simulations corroborated these results and additionally demonstrated that the logit-link CSGLMM$\bm{\lambda}$-g can also empirically recover the corresponding marginal odds ratio estimands (iAOR, cAOR).
Across all our simulation scenarios, the CS-g estimators were unbiased and generally more efficient in the presence of ICS than the existing consistent alternatives, namely the independence estimating equation \cite{wang_two_2022} and basic models with model-robust standardization (MRS) \cite{li_model-robust_2025}.

A practical appeal of the proposed approach is that it preserves the modeling framework that many investigators already prefer.
The basic linear mixed-effects model and GEE with an exchangeable working correlation remain widely used in CRTs despite their potential issues with ICS \cite{wang_two_2022}.
The CS-g approach retains the mixed-effects model or GEE while restoring consistent estimation of the iATE and cATE in P-CRTs with ICS.
Moreover, these estimators are readily implemented with standard statistical software for fitting mixed-effects models and GEEs (e.g., \texttt{lme4} \& \texttt{geepack} packages in \texttt{R}), with g-computation carried out through routine post-estimation prediction (e.g., \texttt{predict()}; Supplementary Appendix \ref{app:example_code}) and inference obtained through a ``leave-one-cluster-out'' jackknife variance estimator.
Furthermore, Theorem \ref{Theorem_g=MRS} implies easy implementation of the described consistent CS-g estimators via the previously developed \texttt{MRStdCRT} package in \texttt{R} \cite{li_model-robust_2025}.
This positions the CS-g approach as a  specific incidence of the general MRS procedure \cite{li_model-robust_2025}. However, CS-g is computationally simpler than MRS, requiring only the simplest g-computation framework to achieve robust consistency, and does not require the more complicated additional weighted cluster-level residual term employed by MRS (Section \ref{sect:g=MRS}). 
We summarize some qualitative comparisons between the proposed CS-g approaches against existing basic models with MRS in Table \ref{tab:summary}.

\begin{table}[htbp]
\centering
\caption{
Qualitative comparisons of the proposed CS-g estimators (CSLMM$\bm{\lambda}$-g, CSGEE$\bm{\lambda}$-g, CSGLMM$\bm{\lambda}$-g)
against estimators derived from existing basic models with model-robust standardization (LMM$\bm{\lambda}$-MRS, GEE$\bm{\lambda}$-MRS, GLMM$\bm{\lambda}$-MRS).
Models are compared in terms of robust consistency, empirical unbiasedness, empirical efficiency, generality, and theoretical simplicity.
}
\label{tab:summary}
\begin{tabular}{|r|c|c|c|c|c|c|}
\hline
& CSLMM$\bm{\lambda}$-g & CSGEE$\bm{\lambda}$-g & CSGLMM$\bm{\lambda}$-g & LMM$\bm{\lambda}$-MRS & GEE$\bm{\lambda}$-MRS  & GLMM$\bm{\lambda}$-MRS \\
\hline
Consistency & \checkmark & \checkmark & & \checkmark & \checkmark & \checkmark \\
\hdashline
Emp. Unbiased & \checkmark & \checkmark & \checkmark & \checkmark & \checkmark & \checkmark\\
\hdashline
Emp. Efficiency & \checkmark & \checkmark & \checkmark & & & \\
\hdashline
Generality & & & & \checkmark & \checkmark & \checkmark \\
\hdashline
Simplicity & \checkmark & \checkmark & \checkmark & & & \\
\hline
\end{tabular}
\end{table}

Conceptually, our results offer a model-robust perspective to CS-g which had been previously suggested for addressing ICS in mixed-effects models \cite{hooper_wood_2026}.
Whereas that prior work emphasized careful specification of the treatment $\times$ cluster-size relationship to model the ICS \cite{hooper_wood_2026}, our derivations show that it is simply the inclusion of the saturated cluster-size terms, rather than their correct functional form, that confers consistency. Therefore, the working cluster-size $\times$ treatment interaction may be arbitrarily misspecified and the resulting g-computation estimators will remain consistent for the iATE and cATE.
This reframes ICS as a feature that is absorbed by saturation rather than one that must be precisely modeled, and clarifies that researchers do not need to commit to a particular parametric form for the cluster-size dependence in order to obtain valid marginal effects.

We observed notable efficiency gains when using the proposed CS-g estimators in our different simulation scenarios.
A downstream benefit of this improved efficiency is more powerful inference for detecting ICS using the test that contrasts the unweighted and inverse cluster-size weighted estimators \cite{li_model-robust_2025}. The CS-g estimators generally attained higher average test statistics and greater power than the corresponding independence estimating equation and basic models with MRS comparisons (Table \ref{tab:ICS_test}). This was also reflected in the case study re-analysis (Table \ref{tab:case_study_comparisons}).
Additional adjustment for other cluster-level covariates can then further improve precision (Supplementary Appendix \ref{app:extended_sim}) while maintaining the derived robust consistency results (Theorem \ref{Theorem_GEE_g}).

Some caveats apply to the binary outcome setting.
Unlike the CSLMM$\bm{\lambda}$-g estimator, the consistency of the CSGEE$\bm{\lambda}$-g estimator relies on one of several supplementary assumptions: (I) a working independence correlation, (II) an identity link with constant working variance, (III) a working mean model restricted to cluster-level covariates, or (IV) a correctly specified mean model (Theorem \ref{Theorem_GEE_g}).
The practical implication of supplementary assumption (II) is that individual-level covariate adjustment under a misspecified exchangeable working correlation with a non-identity link can compromise consistency.
However, we empirically observe that in some data-generating processes, adjustment for individual-level covariates can still yield minimally biased results (Supplementary Appendix \ref{app:extended_sim}).
Still, we do not recommend adjusting for individual-level covariates for non-identity link GEEs with exchangeable working correlation due to the aforementioned theoretically demonstrated lack of model-robustness (Theorem \ref{Theorem_GEE_g}).
Finally, although the CSGLMM$\bm{\lambda}$-g estimators performed well empirically across our simulations, we did not analytically establish their consistency; the marginalization of a non-identity-link generalized linear mixed-effects model over the random effects does not admit the same closed-form algebraic simplifications exploited in the linear and GEE proofs, and a formal derivation is left to future work.

For inference, we favored the ``leave-one-cluster-out'' jackknife variance estimator for its generalizable implementation across statistical software \cite{bell_bias_2002, ford_maintaining_2020}.
Furthermore, with Theorem \ref{Theorem_g=MRS}, we can conveniently use the \texttt{MRStdCRT} package in \texttt{R} to automatically implement these jackknife variance estimators \cite{li_model-robust_2025}.
While Lemma \ref{lemma:variance} establishes the consistency of the cluster-robust sandwich variance estimator, this estimator is known to be biased and to under-cover with small numbers of clusters, even under a $t$-distribution with Satterthwaite degrees of freedom; the jackknife and bias-reduced linearization estimators were accordingly developed to address this \cite{bell_bias_2002}.
In very rare instances, the jackknife variance for the CSGEEw-g failed to converge and the affected replicates were excluded from the simulation results; such instability is consistent with documented sensitivities of the GEE with an exchangeable working correlation in small samples.
We note that some of our replicated analyses may deviate slightly from those originally reported in Li et al. \cite{li_model-robust_2025}.
These differences are partly attributed to our use of a slightly different jackknife variance estimator, which forms the jackknife residuals as differences of the replicate estimates from the full-data point estimate rather than from the average of the replicate estimates; both variants were outlined by Bell \& McCaffrey \cite{bell_bias_2002} and are available in the \texttt{MRStdCRT} package in \texttt{R}.
Additionally, we constructed confidence intervals and p-values using $t$-distributions with a more conservative $m-2$ degrees of freedom suggested by Ford \& Westgate \cite{ford_maintaining_2020}, rather than the $m-1$ degrees of freedom used in Li et al. \cite{li_model-robust_2025}; in practice this distinction is slight and any differences diminish as the number of clusters $m$ increases.

This work has some limitations that motivate future directions.
Our proofs assume that the entire cluster source population is enrolled, whereas accommodating within-cluster sampling may broaden applicability. Still, the described proofs should be easily extendable to accommodate informative enrollment.
Separately, the interpretation of the ICC estimate from the proposed cluster-size saturated models may differ from those commonly reported from the basic linear mixed-effects models and GEEs due to the inclusion of the additional cluster-size terms. In general, ICC estimation in the presence of ICS is still not well understood and should be clarified in future work.
Extensions to the stepped-wedge, crossover, and other multi-period designs are also natural next steps.
In these multi-period CRT designs, more forms of informative sizes may also arise \cite{lee_what_2025}.
Optimistically, the simple parametric nature of our solution in P-CRTs may imply a simple extension to these multi-period CRT designs.

Altogether, our results support cluster-size saturated mixed-effects models and GEEs with g-computation as simple, unbiased, and often more efficient alternatives to existing consistent approaches for targeting interpretable individual and cluster-level estimands in cluster randomized trials with informative cluster sizes.

\section*{acknowledgements}
Research in this article was supported by a Patient-Centered Outcomes Research Institute Award\textsuperscript{\textregistered} (PCORI\textsuperscript{\textregistered} Award ME-2022C2-27676).
The statements presented in this article are solely the responsibility of the authors and do not necessarily represent the official views of the National Institutes of Health, PCORI\textsuperscript{\textregistered}, its Board of Governors, or the Methodology Committee.

\section*{conflict of interest}
The authors declare no potential conflicts of interest with respect to the research, authorship, and/or publication of this article

\section*{Data Availability Statement}
Data sharing is not applicable to this article as no new data were created or analyzed in this study. R codes for the simulations will be deposited to figshare by Wiley.


\printendnotes

\bibliographystyle{wileyNJD-AMA}

\bibliography{references}

\newpage
\appendix

\section{Example code for implementing CS-g}
\label{app:example_code}

We outline some example code illustrating how to implement CS-g estimators in \texttt{R} using standard packages: \texttt{lme4} and \texttt{geepack}. The example code below also uses \texttt{dplyr}, although it's not required in practice.
For the CSLMMw, we rely on the equivalence between a linear mixed-effects model and a corresponding GEE. Cluster size and inverse cluster sizes are included as continuous variable: ``clusterSize'' and ``clusterSizeInverse'', respectively.
\begin{enumerate}
    \item Specify the model as a cluster-size saturated mixed-effects model or GEE for continuous outcomes:
    \begin{itemize}
        \item \textbf{CSLMM:} 
        \begin{itemize}
           \item model <- lme4::lmer(outcome $\sim$ treatment*clusterSize + (1|cluster), data)
        \end{itemize}
        \item \textbf{CSLMMw:}
        \begin{itemize}
            \item modelw <- geepack::geeglm(outcome $\sim$ treatment*clusterSize, id = cluster, family = gaussian(link = ``identity''), corstr = "exchangeable", weights=clusterSizeInverse, data)
        \end{itemize}
    \end{itemize}
    or for binary outcomes:
    \begin{itemize}
        \item \textbf{CSGEE:}
        \begin{itemize}
            \item model <- geepack::geeglm(outcome $\sim$ treatment*clusterSize, id = cluster, binomial(link = ``logit''), corstr = "exchangeable", data)
        \end{itemize}
        \item \textbf{CSGEEw:}
        \begin{itemize}
            \item modelw <- geepack::geeglm(outcome $\sim$ treatment*clusterSize, id = cluster, binomial(link = ``logit''), corstr = "exchangeable", weights=clusterSizeInverse, data)
        \end{itemize}
        \end{itemize}
    \item Implement g-computation:
    \begin{itemize}
        \item First, create corresponding datasets where all individuals are under treatment or control:
        \begin{itemize}
            \item data1 <- data \%>\% mutate(treatment = 1)
            \item data0 <- data \%>\% mutate(treatment = 0)
        \end{itemize}
        \item Predict outcomes under treatment or control with unweighted and weighted models
        \begin{itemize}
            \item predOutcome1 <- predict(model, data1, type=``response'')
            \item predOutcome0 <- predict(model, data0, type=``response'')
            \item predOutcome1w <- predict(modelw, data1, type=``response'')
            \item predOutcome0w <- predict(modelw, data0, type=``response'')
        \end{itemize}
    \end{itemize}
    \item Summarize into the desired effect measure:
    \begin{itemize}
        \item Estimate the iATE \& cATE:
        \begin{itemize}
            \item iATE <- mean(predOutcome1 - predOutcome0)
            \item cATE <- weighted.mean(predOutcome1w - predOutcome0w, data\$clusterSizeInverse)
        \end{itemize}
        \item Or estimate the iAOR or cAOR:
        \begin{itemize}
            \item iAOR <- (mean(predOutcome1)/(1-mean(predOutcome1))) / (mean(predOutcome0)/(1-mean(predOutcome0)))
            \item cAOR <- (weighted.mean(predOutcome1, data\$clusterSizeInverse) / (1-weighted.mean(predOutcome1, data\$clusterSizeInverse))) / (weighted.mean(predOutcome0, data\$clusterSizeInverse) / (1-weighted.mean(predOutcome0, data\$clusterSizeInverse)))
        \end{itemize}
    \end{itemize}
\end{enumerate}

Subsequent calculation of the ``leave-one-cluster-out'' jackknife variance then proceeds as described in Bell \& McCaffrey \cite{bell_bias_2002}. This process iterates through the above code and removes a different cluster in each re-sample.

Calculation of the test-statistic for detecting ICS is easily implemented by extending the above code and implementing a corresponding jackknife variance estimator, following its description in Li et al. \cite{li_model-robust_2025}.

With our demonstrated exact finite sample equivalence between the consistent cluster-size saturated models with g-computation (CS-g) and cluster-size saturated models with model-robust standardization (CS-MRS) (Section \ref{sect:g=MRS}). 
The described consistent CS-g estimators, the different forms of jackknife variance estimators, and the test-statistic for detecting ICS can therefore be easily implemented within the \texttt{MRStdCRT} package in \texttt{R}.

\section{Extended Proofs to Lemmas \ref{lemma:variance} - \ref{lemma:g} \& Theorems \ref{Theorem_g} - \ref{Theorem_g=MRS}}
\label{app:proofs}

Proofs for Lemmas \ref{lemma:variance} - \ref{lemma:g} \& Theorems \ref{Theorem_g} - \ref{Theorem_g=MRS} are included in the order that they appear in the main manuscript.

\subsection{Proof of Lemma \ref{lemma:variance}}
The  three stated regularity conditions in Lemma \ref{lemma:variance} are standard assumptions in an M-estimation framework \cite{tsiatis_semiparametric_2006,van_der_vaart_asymptotic_1998,ross_m-estimation_2024} to ensure the estimating function $\bm{\psi}$ (equivalent to the likelihood score function) is well-behaved to prove the asymptotic results.
Of note, condition (1) does not imply any component of the working model is correctly specified. Instead, it solely requires the uniqueness of maxima in the maximum likelihood or maximum quasi-likelihood estimation.
This can be achieved by carefully designing $\bm{\psi}$ and restricting the parameter space to rule out degenerative solutions.
As a specific example, if $\bm{\psi}$ is the estimating equation for linear regression (based on ordinary least squares), then condition (1) is equivalent to the invertibility of the covariance matrix of covariates.

We use the notation $o_p(1)$ to denote a sequence of random vectors that conveges to zero in probability, and $O_p(1)$ to denote a sequence that is bounded in probability, as per Section 2.2 of \cite{van_der_vaart_asymptotic_1998}.
An asymptotically linear estimator $\hat{\bm{\theta}}$, excluding nuisance parameters (as described in Section 3 of \cite{tsiatis_semiparametric_2006}), can be uniquely characterized by its influence function, as demonstrated in Theorem 3.1 of \cite{tsiatis_semiparametric_2006}. The influence function captures the influence of the $i$th observation on $\hat{\bm{\theta}}$ and is denoted as $\text{IF}(\bm{O}_i,\underline{\bm{\theta}}) = -\left( E\left[\frac{d\bm{\psi}(\bm{O},\bm{\theta})}{d\bm{\theta}^{\top}} \mid_{\bm{\theta}=\underline{\bm{\theta}}}\right]^{-1} \bm{\psi}(\bm{O},\underline{\bm{\theta}}) \right)$, as per Equation 3.6 of \cite{tsiatis_semiparametric_2006}.

The proof proceeds by using ``classical conditions'' for asymptotic normality of M-estimators, as per Section 5.6, and largely follows Theorem 5.41, of \cite{van_der_vaart_asymptotic_1998}.
By condition (2) for $\bm{\psi}$, Example 19.8 of \cite{van_der_vaart_asymptotic_1998} implies that $\{\bm{\psi}(\bm{O},\bm{\theta}) : \bm{\theta} \in \bm{\Theta}\}$ is P-Glivenko-Cantelli (Glivenko-Cantelli in probability).
Then with condition (1), Theorem 5.9 of \cite{van_der_vaart_asymptotic_1998} shows that $\hat{\bm{\theta}} \xrightarrow{P} \underline{\bm{\theta}}$.
Next, we apply Theorem 5.41 of \cite{van_der_vaart_asymptotic_1998} to obtain asymptotic normality, for which our assumptions on $\bm{\psi}$ ensure all conditions needed are satisfied. Then we have
\begin{equation}
    m^{1/2}(\hat{\bm{\theta}} - \underline{\bm{\theta}}) = m^{-1/2} \left(\sum_{i=1}^{m} \text{IF}(\bm{O}_i, \underline{\bm{\theta}})\right) + o_p(1) \,,
\end{equation}
as per Theorem 5.41 of \cite{van_der_vaart_asymptotic_1998} and equation 3.1 of \cite{tsiatis_semiparametric_2006}, which implies the desired asymptotic normality by the Central Limit Theorem.
Specifically, by the central limit theorem, $m^{1/2} \left(\sum_{i=1}^{m} \text{IF}(\bm{O}_i, \underline{\bm{\theta}})\right) \xrightarrow{D} N(0, E[\text{IF}(\bm{O},\underline{\bm{\theta}}) \text{IF}(\bm{O},\underline{\bm{\theta}})^\top])$ and by Slutsky's theorem,
\begin{equation}
    m^{1/2}(\hat{\bm{\theta}} - \underline{\bm{\theta}}) \xrightarrow{D} N\left( 0, E[\text{IF}(\bm{O},\underline{\bm{\theta}}) \text{IF}(\bm{O},\underline{\bm{\theta}})^\top] \right)
\end{equation}
hence $\textbf{V} = E[\text{IF}(\bm{O},\underline{\bm{\theta}}) \text{IF}(\bm{O},\underline{\bm{\theta}})^\top]$, as per Equation 3.7 of \cite{tsiatis_semiparametric_2006}.

We next prove the consistency of the sandwich variance estimator. First, we prove that
\[
    m^{-1}\sum_{i=1}^{m} \frac{d\bm{\psi}(\bm{O}_i,\bm{\theta})}{d\bm{\theta}^{\top}} \mid_{\bm{\theta}=\hat{\bm{\theta}}}
\xrightarrow{P} E\left[\frac{d\bm{\psi}(\bm{O},\bm{\theta})}{d\bm{\theta}^{\top}} \mid_{\bm{\theta}=\underline{\bm{\theta}}}\right] \,.
\]
Denoting $\dot{\bm{\psi}}_{ik}(\hat{\bm{\theta}})$ as the transpose of the $k$th row of $\frac{d\bm{\psi}(\bm{O}_i,\bm{\theta})}{d\bm{\theta}^{\top}} \mid_{\bm{\theta}=\hat{\bm{\theta}}}$, and $\ddot{\bm{\psi}}_{ik}$ being the derivative of $\dot{\bm{\psi}}_{ik}$, we apply the multivariate Taylor expansion to get
\[
    m^{-1} \sum_{i=1}^{m} \dot{\bm{\psi}}_{ik}(\hat{\bm{\theta}}) - m^{-1} \sum_{i=1}^{m} \dot{\bm{\psi}}_{ik}(\underline{\bm{\theta}}) 
    = m^{-1} \sum_{i=1}^{m} \left( \dot{\bm{\psi}}_{ik}(\hat{\bm{\theta}}) - \dot{\bm{\psi}}_{ik}(\underline{\bm{\theta}}) \right)
    = m^{-1} \left(  \sum_{i=1}^{m} \ddot{\bm{\psi}}_{ik}(\tilde{\bm{\theta}}) \right)(\hat{\bm{\theta}}-\underline{\bm{\theta}})
\]
for some $\tilde{\bm{\theta}}$ on the line segment between $\hat{\bm{\theta}}$ and $\underline{\bm{\theta}}$.
By condition 2 and $\hat{\bm{\theta}} \xrightarrow{P} \bm{\theta}$, we have $m^{-1} \sum_{i=1}^{m} \ddot{\bm{\psi}}_{ik}(\tilde{\bm{\theta}}) =O_p(1)$.
As a result, $\hat{\bm{\theta}}-\bm{\theta} = o_p(1)$ implies that $m^{-1} \sum_{i=1}^{m} \dot{\bm{\psi}}_{ik}(\hat{\bm{\theta}}) - m^{-1} \sum_{i=1}^{m} \dot{\bm{\psi}}_{ik}(\underline{\bm{\theta}})  = o_p(1)$.
Then the first step is completed by the fact that
\[
    m^{-1} \sum_{i=1}^{m} \dot{\bm{\psi}}_{ik}(\underline{\bm{\theta}}) = E[\dot{\bm{\psi}}_{ik}(\underline{\bm{\theta}})] + o_p(1)
\]
which results from the Law of Large Numbers and condition 2.
Next we prove
\[
    m^{-1} \sum_{i=1}^{m} \bm{\psi}(\bm{O}_i, \hat{\bm{\theta}}) \bm{\psi}(\bm{O}_i, \hat{\bm{\theta}})^\top 
    \xrightarrow{P} E[\bm{\psi}(\bm{O}_i, \underline{\bm{\theta}}) \bm{\psi}(\bm{O}_i, \underline{\bm{\theta}})^\top ]
\]
following a similar procedure to the prior step. Letting $\bm{\psi}_{ik}(\bm{\theta})$ be the $k$th entry of $\bm{\psi}(\bm{O}_i,\bm{\theta})$, we apply the multivariate Taylor expansion and get
\[
\begin{split}
    m^{-1} &\sum_{i=1}^{m} \bm{\psi}_{ik}(\hat{\bm{\theta}}) \bm{\psi}(\bm{O}_i, \hat{\bm{\theta}}) -  m^{-1} \sum_{i=1}^{m} \bm{\psi}_{ik}(\underline{\bm{\theta}}) \bm{\psi}(\bm{O}_i, \underline{\bm{\theta}}) \\
    &= m^{-1} \left( \sum_{i=1}^{m} \left[ 
    \bm{\psi}(\bm{O}_i,\tilde{\bm{\theta}})\dot{\bm{\psi}}_{ik}(\tilde{\bm{\theta}})^\top + \bm{\psi}_{ik}(\tilde{\bm{\theta}}) \frac{d\bm{\psi}(\bm{O}_i,\bm{\theta})}{d\bm{\theta}^{\top}} \mid_{\bm{\theta}=\tilde{\bm{\theta}}} 
    \right] \right) (\hat{\bm{\theta}}-\underline{\bm{\theta}}) \\
    &= O_p(1)o_p(1)
\end{split}
\]
which, combined with the Law of Large Numbers on $m^{-1} \sum_{i=1}^{m} \bm{\psi}_{ik}(\underline{\bm{\theta}})\bm{\psi}(\bm{O}_i,\underline{\bm{\theta}})$, implies the desired result in this step.
The sandwich variance estimator can be written in terms of influence functions while substituting $\hat{\bm{\theta}}$ for the unknown $\underline{\bm{\theta}}$, as defined in equation 3.10 of \cite{tsiatis_semiparametric_2006}.
\begin{align}
    m^{-1} &\sum_{i=1}^{m} \widehat{\text{IF}}(\bm{O}_i,\hat{\bm{\theta}}) \widehat{\text{IF}}(\bm{O}_i,\hat{\bm{\theta}})^\top \\
    &=\left(
        m^{-1} \sum_{i=1}^{m} \frac{d\bm{\psi}(\bm{O}_i,\bm{\theta})}{d\bm{\theta}^{\top}} \mid_{\bm{\theta}=\hat{\bm{\theta}}}
    \right)^{-1}
    \left(
        m^{-1} \sum_{i=1}^{m} \bm{\psi}(\bm{O}_i,\hat{\bm{\theta}}) \bm{\psi}(\bm{O}_i,\hat{\bm{\theta}})^\top
    \right)
    \left(
        m^{-1} \sum_{i=1}^{m} \frac{d\bm{\psi}(\bm{O}_i,\bm{\theta})}{d\bm{\theta}^{\top}} \mid_{\bm{\theta}=\hat{\bm{\theta}}}
    \right)^{-1\top} \notag \\
\intertext{Finally by the Continuous Mapping Theorem described in Theorem 2.3 of \cite{van_der_vaart_asymptotic_1998}, and the above derivations, we can show the sandwich variance estimator}
    &=\left(
        E\left[\frac{d\bm{\psi}(\bm{O},\bm{\theta})}{d\bm{\theta}^{\top}} \mid_{\bm{\theta}=\underline{\bm{\theta}}}\right] + o_p(1)
    \right)^{-1}
    \left(
        E\left[ \bm{\psi}(\bm{O},\underline{\bm{\theta}}) \bm{\psi}(\bm{O},\underline{\bm{\theta}})^\top \right] + o_p(1)
    \right)
    \left(
        E\left[\frac{d\bm{\psi}(\bm{O},\bm{\theta})}{d\bm{\theta}^{\top}} \mid_{\bm{\theta}=\underline{\bm{\theta}}}\right] + o_p(1)
    \right)^{-1\top} \notag \\
    &= \textbf{V} + o_p(1) \,.
\end{align}
Altogether, the sandwich variance estimator is consistent for the true variance \textbf{V}. $\square$

\subsection{Proof of Theorem \ref{Theorem_g}.a.}
Denote $\bm{\theta}=(\Delta_{CSLMM-g}, \bm{\beta}, \tau^2, \sigma^2)^{\top} \in \mathbb{R}^{7+p}$ as the vector of unknown parameters to be estimated by M-estimation.
$\Delta_{CSLMM-g}$ is the CSLMM-g estimator that we will demonstrate is consistent for the iATE ($\Delta_{iATE}$) in P-CRTs with potential ICS.
Based on the observed data, the log-likelihood function given $\{Z_i, \bm{X}_i, N_i\}$ is
\[
\begin{split}
    & l(\bm{\theta}; \{\bm{Y}_i\}_{i=1}^{m} | \{Z_i, \bm{X}_i,  N_i\}_{i=1}^{m}) \\
    & = C - \frac{1}{2} \sum_{i=1}^{m} \left[ log(| \bm{\Sigma}_i|) + (\bm{Y}_i - \bm{Q}_i\bm{\beta})^\top  \bm{\Sigma}_i^{-1} (\bm{Y}_i - \bm{Q}_i\bm{\beta}) \right]
\end{split}
\]
where $C$ is a constant independent of  parameters $\bm{\theta}$.
The estimators are then the solution to the estimating equations $\sum_{i=1}^m \bm{\psi}(\bm{O}_i, \bm{\theta})$, where
\[
\begin{split}
    \bm{\psi}(\bm{O}_i, \bm{\theta})
    & = 
        \left(
        \begin{gathered}
            \left(\frac{\sum_{s=1}^{m}N_s}{m}\right) \Delta_{CSLMM-g} - \bm{1}_{N_i}^\top(\bm{\mu}_i(1) - \bm{\mu}_i(0)) \\
            \bm{Q}_i^\top \bm{V}_i (\bm{Y}_i - \bm{Q}_i\bm{\beta}) \\
            -tr(\bm{V}_i) + (\bm{Y}_i - \bm{Q}_i\bm{\beta})^\top \bm{V}^2_i (\bm{Y}_i - \bm{Q}_i\bm{\beta}) \\
            -\bm{1}_{N_i}^\top \bm{V}_i \bm{1}_{N_i} + (\bm{Y}_i - \bm{Q}_i\bm{\beta})^\top \bm{V}_i ( \bm{1}_{N_i} \bm{1}_{N_i}^\top) \bm{V}_i  (\bm{Y}_i - \bm{Q}_i\bm{\beta})
        \end{gathered}
        \right)
\end{split}
\]
\sloppy
where
$\bm{\mu}_i(a) = \bm{Q}_i(a) \bm{\beta} = \left(\beta_{0}  + a \beta_Z + a N_i  \beta_{ZN} + N_i \beta_{N} + \bm{\beta}_X^\top \bm{X}_{ik}\right)_{k=1,...,N_i} \in \mathbb{R}^{N_i}$, $a\in\{0,1\}$,
$\bm{V}_i = \bm{\Sigma}_i^{-1} \in \mathbb{R}^{N_i \times N_i}$, and $tr(\bm{V}_i)$ is the trace of $\bm{V}_i$.
With $\rho=\frac{\tau^2}{\tau^2 + \sigma^2}$ being the ICC, we can then easily demonstrate that $\bm{V}_i = \left(\frac{1}{\tau^2 + \sigma^2}\right)\left(\frac{1}{1-\rho}\right)\left(\textbf{I}_{N_i} - \bm{1}_{N_i} \bm{1}_{N_i}^\top (\rho/[1+(N_i-1)\rho])\right)$.
The maximum likelihood estimator for $\bm{\theta}$ is defined as a solution to the estimating equation
\[
    \sum_{i=1}^{m} \bm{\psi}(\bm{O}_i;\bm{\theta})=0 \,.
\]
We hence define the estimating function as
\begin{equation}
    \bm{\psi}(\bm{O};\bm{\theta}) =  \left(
        \begin{gathered}
            E[N] \Delta_{CSLMM-g} - \bm{1}_{N}^\top(\bm{\mu}(1) - \bm{\mu}(0))  \\
            \bm{Q}^\top \bm{V} (\bm{Y} - \bm{Q}\bm{\beta}) \\
            -tr(\bm{V}) + (\bm{Y} - \bm{Q}\bm{\beta})^\top \bm{V}^2(\bm{Y} - \bm{Q}\bm{\beta}) \\
            -\bm{1}_{N}^\top \bm{V} \bm{1}_{N} + (\bm{Y} - \bm{Q}\bm{\beta})^\top \bm{V} ( \bm{1}_{N} \bm{1}_{N}^\top) \bm{V}  (\bm{Y} - \bm{Q}\bm{\beta})
        \end{gathered}
        \right)
\end{equation}
Recall that subscript $i$ is omitted when taking the expectation with respect to distribution $\mathcal{P}$.
For the estimating equation $\bm{\psi}$, we prove the convergence and asymptotic normality of $\hat{\bm{\theta}}$ by applying Lemma \ref{lemma:variance}. In Lemma \ref{lemma:variance}, the conditions for $\bm{\psi}$ are assumed in the main paper as regularity conditions, which implies the desired results. The consistency of the sandwich variance estimators is then also implied.

Denoting $\underline{\bm{\theta}}=(\underline{\Delta}_{CSLMM-g}, \underline{\bm{\beta}}_0, \underline{\beta}_Z, \underline{\beta}_{ZN}, \underline{\beta}_{N}, \underline{\bm{\beta}}_{X}^\top, \underline{\tau}^2, \underline{\sigma^2})^{\top}$
as the solution to $E[\bm{\psi}(\bm{O};\bm{\theta})]=0$, we next prove $\underline{\Delta}_{CSLMM-g} = \Delta_{iATE}$ to imply robust consistency of this CSLMM-g estimator.
To proceed, the second through fifth entries of $E[\bm{\psi}(\bm{O};\bm{\theta})]=0$ are
\begin{gather}
\label{app:eq:baseline}
    E[\bm{1}_N^\top \underline{\bm{V}}(\bm{Y} - \bm{Q}\underline{\bm{\beta}})] = 0 \\
\label{app:eq:treatment}
    E[I\{Z=1\}\bm{1}_N^\top \underline{\bm{V}}(\bm{Y} - \bm{Q}\underline{\bm{\beta}})] = 0 \\
\label{app:eq:treatment_N}
    E[I\{Z=1\} N \bm{1}_N^\top \underline{\bm{V}}(\bm{Y} - \bm{Q}\underline{\bm{\beta}})] = 0 \\
\label{app:eq:N}
    E[N\bm{1}_N^\top \underline{\bm{V}}(\bm{Y} - \bm{Q}\underline{\bm{\beta}})] = 0
\end{gather}
corresponding to the score functions for the model intercept $\beta_0$ (Equation \ref{app:eq:baseline}), treatment effect coefficient $\beta_Z$ (Equation \ref{app:eq:treatment}), treatment $\times$ cluster-size interaction $\beta_{ZN}$ (Equation \ref{app:eq:treatment_N}), and cluster-size main effect $\beta_N$ (Equation \ref{app:eq:N}), respectively. $\underline{\bm{V}}$ is equal to $\bm{V}$ with $(\tau^2, \sigma^2, \rho)$ replaced by $(\underline{\tau}^2, \underline{\sigma}^2, 
\underline{\rho})$.

It is then trivial to algebraically demonstrate that the above estimating equations are respectively equal to
\[
\begin{split}
E\left[\left(\frac{1}{1+(N-1)\underline{\rho}}\right) \bm{1}_N^\top (\bm{Y} - \underline{\bm{\mu}})\right] &= 0 \,, \\
E\left[I\{Z=1\} \left(\frac{1}{1+(N-1)\underline{\rho}}\right) \bm{1}_N^\top (\bm{Y} - \underline{\bm{\mu}})\right] &= 0 \,, \\
E\left[I\{Z=1\} N \left(\frac{1}{1+(N-1)\underline{\rho}}\right) \bm{1}_N^\top (\bm{Y} - \underline{\bm{\mu}})\right] &= 0 \,, \\
E\left[ N \left(\frac{1}{1+(N-1)\underline{\rho}}\right) \bm{1}_N^\top (\bm{Y} - \underline{\bm{\mu}})\right] &= 0 \,,
\end{split}
\]
where $\underline{\bm{\mu}} = \bm{Q}\underline{\bm{\beta}}$.
Recall that $\pi=P(Z_i=1)=E[Z_i]$ with Assumption A2 (Cluster randomization).
Then, taking Equation (\ref{app:eq:treatment_N}) and subtracting it by $\pi \times$ Equation (\ref{app:eq:N}) returns
\[
    E\left[(I\{Z=1\} N - \pi N) \left(\frac{1}{1+(N-1)\underline{\rho}}\right) \bm{1}_N^\top (\bm{Y} - \underline{\bm{\mu}})\right] = 0 \,.
\]
Given $\bm{Y}_i = I\{Z_i=1\}\left[\bm{Y}_i(1)-\bm{Y}_i(0)\right] + \bm{Y}_i(0)$ and $\underline{\bm{\mu}}_i = I\{Z_i=1\}\left[\underline{\bm{\mu}}_i(1)-\underline{\bm{\mu}}_i(0)\right] + \underline{\bm{\mu}}_i(0)$ in a P-CRT, we then have the following equality
\[
\begin{split}
    E&\left[(I\{Z=1\} N - \pi N) \left(\frac{1}{1+(N-1)\underline{\rho}}\right) \bm{1}_N^\top I\{Z_i=1\} \left\{\left[\bm{Y}(1)-\bm{Y}(0)\right] - \left[\underline{\bm{\mu}}(1) - \underline{\bm{\mu}}(0)]\right]\right\} \right] \\
    &+ E\left[(I\{Z=1\} N - \pi N) \left(\frac{1}{1+(N-1)\underline{\rho}}\right) \bm{1}_N^\top (\bm{Y}(0) - \underline{\bm{\mu}}(0))\right] \\
    &=0
\end{split}
\]
Notably, the second component of the above summation can be shown to be equal to 0
\[
\begin{split}
    E&\left[(I\{Z=1\} N - \pi N) \left(\frac{1}{1+(N-1)\underline{\rho}}\right) \bm{1}_N^\top (\bm{Y}(0) - \underline{\bm{\mu}}(0))\right] \\
    &= E\left[E\left[(I\{Z=1\} N - \pi N) \left(\frac{1}{1+(N-1)\underline{\rho}}\right) \bm{1}_N^\top (\bm{Y}(0) - \underline{\bm{\mu}}(0)) | N\right]\right] \\
    &= E\left[(E\left[I\{Z=1\}|N\right] N - \pi N) \left(\frac{1}{1+(N-1)\underline{\rho}}\right) \bm{1}_N^\top E\left[\bm{Y}(0) - \underline{\bm{\mu}}(0) | N\right]\right] \\
    &= 0
\end{split}
\]
where in the second equality, $E\left[I\{Z=1\}|N\right]=E\left[I\{Z=1\}\right] = \pi$ by Assumption A2 (Cluster randomization).
Altogether, we demonstrate that Equation (\ref{app:eq:treatment_N}) can be transformed to
\[
\begin{split}
    E&\left[(I\{Z=1\} N - \pi N) \left(\frac{1}{1+(N-1)\underline{\rho}}\right) \bm{1}_N^\top I\{Z_i=1\} \left\{\left[\bm{Y}(1)-\bm{Y}(0)\right] - \left[\underline{\bm{\mu}}(1) - \underline{\bm{\mu}}(0)]\right]\right\} \right] \\
    &= E\left[I\{Z=1\} N(1 - \pi) \left(\frac{1}{1+(N-1)\underline{\rho}}\right) \bm{1}_N^\top \left\{\left[\bm{Y}(1)-\bm{Y}(0)\right] - \left[\underline{\bm{\mu}}(1) - \underline{\bm{\mu}}(0)]\right]\right\} \right] \\
    &= E\left[E\left[I\{Z=1\}|N\right] N \left(\frac{1}{1+(N-1)\underline{\rho}}\right) \bm{1}_N^\top \left\{\left[\bm{Y}(1)-\bm{Y}(0)\right] - \left[\underline{\bm{\mu}}(1) - \underline{\bm{\mu}}(0)]\right]\right\} \right] \\
    &= E\left[N \left(\frac{1}{1+(N-1)\underline{\rho}}\right) \bm{1}_N^\top \left\{\left[\bm{Y}(1)-\bm{Y}(0)\right] - \left[\underline{\bm{\mu}}(1) - \underline{\bm{\mu}}(0)]\right]\right\}\right] \\
    &= 0
\end{split}
\]
where the third equality, as before, also results from Assumption A2 (Cluster randomization).

Similarly, we can demonstrate that Equation (\ref{app:eq:treatment}) can likewise be subtracted by $\pi \times$ Equation (\ref{app:eq:baseline}) and transformed to
\[
\begin{split}
    E&\left[(I\{Z=1\} - \pi) \left(\frac{1}{1+(N-1)\underline{\rho}}\right) \bm{1}_N^\top I\{Z_i=1\} \left\{\left[\bm{Y}(1)-\bm{Y}(0)\right] - \left[\underline{\bm{\mu}}(1) - \underline{\bm{\mu}}(0)]\right]\right\} \right] \\
    &= E\left[I\{Z=1\} (1 - \pi) \left(\frac{1}{1+(N-1)\underline{\rho}}\right) \bm{1}_N^\top \left\{\left[\bm{Y}(1)-\bm{Y}(0)\right] - \left[\underline{\bm{\mu}}(1) - \underline{\bm{\mu}}(0)]\right]\right\} \right] \\
    &= E\left[E\left[I\{Z=1\}|N\right] \left(\frac{1}{1+(N-1)\underline{\rho}}\right) \bm{1}_N^\top \left\{\left[\bm{Y}(1)-\bm{Y}(0)\right] - \left[\underline{\bm{\mu}}(1) - \underline{\bm{\mu}}(0)]\right]\right\} \right] \\
    &= E\left[\left(\frac{1}{1+(N-1)\underline{\rho}}\right) \bm{1}_N^\top \left\{\left[\bm{Y}(1)-\bm{Y}(0)\right] - \left[\underline{\bm{\mu}}(1) - \underline{\bm{\mu}}(0)]\right]\right\}\right] \\
    &= 0 \,.
\end{split}
\]

We then have the following equalities
\begin{equation}
\label{app:eq.ee.ZN}
    E\left[N \left(\frac{1}{1+(N-1)\underline{\rho}}\right) \bm{1}_N^\top \left\{\left[\bm{Y}(1)-\bm{Y}(0)\right] - \left[\underline{\bm{\mu}}(1) - \underline{\bm{\mu}}(0)]\right]\right\}\right] = 0
\end{equation}
and
\begin{equation}
\label{app:eq.ee.Z}
    E\left[\left(\frac{1}{1+(N-1)\underline{\rho}}\right) \bm{1}_N^\top \left\{\left[\bm{Y}(1)-\bm{Y}(0)\right] - \left[\underline{\bm{\mu}}(1) - \underline{\bm{\mu}}(0)]\right]\right\}\right] = 0 \,.
\end{equation}
We can then demonstrate that $\Delta_{iATE} = E[\bm{1}_N^\top \{\bm{Y}(1) - \bm{Y}(0)\}]/E[N] = E[\bm{1}_N^\top \{\underline{\bm{\mu}}(1) - \underline{\bm{\mu}}(0)\}]/E[N] = \underline{\Delta}_{CSLMM-g}$ and complete the proof of consistency. For this statement to be true, then $E[\bm{1}_N^\top \{[\bm{Y}(1) - \bm{Y}(0) ]- [\underline{\bm{\mu}}(1) - \underline{\bm{\mu}}(0)]\} ] = 0$ needs to be true.
We prove this is the case with the previously proven results in Equations (\ref{app:eq.ee.ZN}) and (\ref{app:eq.ee.Z})
\[
\begin{split}
    E&[\bm{1}_N^\top \{[\bm{Y}(1) - \bm{Y}(0) ]- [\underline{\bm{\mu}}(1) - \underline{\bm{\mu}}(0)]\}] \\
    &= E\left[ \left(\frac{1+(N-1)\underline{\rho}}{1+(N-1)\underline{\rho}}\right) \bm{1}_N^\top \{[\bm{Y}(1) - \bm{Y}(0) ]- [\underline{\bm{\mu}}(1) - \underline{\bm{\mu}}(0)]\}\right] \\
    &\begin{aligned}
            = E&\left[ \left(\frac{1}{1+(N-1)\underline{\rho}}\right) \bm{1}_N^\top \{[\bm{Y}(1) - \bm{Y}(0) ]- [\underline{\bm{\mu}}(1) - \underline{\bm{\mu}}(0)]\}\right] \\
            &+ \underline{\rho}\left( E\left[ \left(\frac{N}{1+(N-1)\underline{\rho}}\right) \bm{1}_N^\top \{[\bm{Y}(1) - \bm{Y}(0) ]- [\underline{\bm{\mu}}(1) - \underline{\bm{\mu}}(0)]\}\right] \right) \\
            &- \underline{\rho}\left( E\left[ \left(\frac{1}{1+(N-1)\underline{\rho}}\right) \bm{1}_N^\top \{[\bm{Y}(1) - \bm{Y}(0) ]- [\underline{\bm{\mu}}(1) - \underline{\bm{\mu}}(0)]\}\right] \right)
    \end{aligned}\\
    &=0 \,.
\end{split}
\]
In contrast to the basic linear mixed-effects model (Equation \ref{eq:ME_basic}) and models including only a cluster-size main effect, it is the presence of both the treatment $\times$ cluster-size interaction $\beta_{ZN}$ (Equation \ref{app:eq:treatment_N}) and cluster-size main effect $\beta_N$ (Equation \ref{app:eq:N}) which produces Equation (\ref{app:eq.ee.ZN}) in the CSLMM-g estimator to yield the final result:
\[
\begin{split}
   \underline{\Delta}_{CSLMM-g} &= \frac{E[\bm{1}_N^\top \{\underline{\bm{\mu}}(1) - \underline{\bm{\mu}}(0)\}]}{E[N]} \\
    &= \frac{E[\bm{1}_N^\top \{\bm{Y}(1) - \bm{Y}(0)\}]}{E[N]} = \Delta_{iATE} \,.
\end{split}
\]
completing the proof of consistency. $\square$

\subsection{Proof of Theorem \ref{Theorem_g}.b.}

We denote $\bm{\theta}=(\Delta_{CSLMMw-g}, \bm{\beta}, \tau^2, \sigma^2)^{\top} \in \mathbb{R}^{7+p}$ as the vector of unknown parameters to be estimated by M-estimation.
The estimators to target the cATE are then the solution to the following estimating equations $\sum_{i=1}^m \bm{\psi}(\bm{O}_i, \bm{\theta})$, where
\[
\begin{split}
    \bm{\psi}(\bm{O}_i, \bm{\theta})
    & = 
        \left(
        \begin{gathered}
            \Delta_{CSLMMw-g} - \frac{1}{N_i}\bm{1}_{N_i}^\top(\bm{\mu}_i(1) - \bm{\mu}_i(0)) \\
            \frac{1}{N_i}\bm{Q}_i^\top \bm{V}_i (\bm{Y}_i - \bm{Q}_i\bm{\beta}) \\
            -tr(\bm{V}_i) + (\bm{Y}_i - \bm{Q}_i\bm{\beta})^\top \bm{V}^2_i (\bm{Y}_i - \bm{Q}_i\bm{\beta}) \\
            -\bm{1}_{N_i}^\top \bm{V}_i \bm{1}_{N_i} + (\bm{Y}_i - \bm{Q}_i\bm{\beta})^\top \bm{V}_i ( \bm{1}_{N_i} \bm{1}_{N_i}^\top) \bm{V}_i  (\bm{Y}_i - \bm{Q}_i\bm{\beta})
        \end{gathered}
        \right)
\end{split}
\]
The maximum likelihood estimator for $\bm{\theta}$ is defined as a solution to the estimating equation
\[
    \sum_{i=1}^{m} \bm{\psi}(\bm{O}_i;\bm{\theta})=0 \,.
\]
We hence define the estimating function as
\begin{equation}
    \bm{\psi}(\bm{O};\bm{\theta}) =  \left(
        \begin{gathered}
            \Delta_{CSLMMw-g} - \frac{1}{N}\bm{1}_{N}^\top(\bm{\mu}(1) - \bm{\mu}(0))  \\
            \frac{1}{N}\bm{Q}^\top \bm{V} (\bm{Y} - \bm{Q}\bm{\beta}) \\
            -tr(\bm{V}) + (\bm{Y} - \bm{Q}\bm{\beta})^\top \bm{V}^2(\bm{Y} - \bm{Q}\bm{\beta}) \\
            -\bm{1}_{N}^\top \bm{V} \bm{1}_{N} + (\bm{Y} - \bm{Q}\bm{\beta})^\top \bm{V} ( \bm{1}_{N} \bm{1}_{N}^\top) \bm{V}  (\bm{Y} - \bm{Q}\bm{\beta})
        \end{gathered}
        \right) \,.
\end{equation}

The second through fifth entries of $E[\bm{\psi}(\bm{O};\bm{\theta})]=0$ are
\begin{gather}
\label{app:eq:baseline.c}
    E\left[\frac{1}{N}\bm{1}_N^\top \underline{\bm{V}}(\bm{Y} - \bm{Q}\underline{\bm{\beta}})\right] = 0 \\
\label{app:eq:treatment.c}
    E\left[\frac{1}{N}I\{Z=1\}\bm{1}_N^\top \underline{\bm{V}}(\bm{Y} - \bm{Q}\underline{\bm{\beta}})\right] = 0 \\
\label{app:eq:treatment_N.c}
    E\left[I\{Z=1\} \bm{1}_N^\top \underline{\bm{V}}(\bm{Y} - \bm{Q}\underline{\bm{\beta}})\right] = 0 \\
\label{app:eq:N.c}
    E\left[\bm{1}_N^\top \underline{\bm{V}}(\bm{Y} - \bm{Q}\underline{\bm{\beta}})\right] = 0
\end{gather}
corresponding to the score functions for the model intercept $\beta_0$ (Equation \ref{app:eq:baseline.c}), treatment effect coefficient $\beta_Z$ (Equation \ref{app:eq:treatment.c}), treatment $\times$ cluster-size interaction $\beta_{ZN}$ (Equation \ref{app:eq:treatment_N.c}), and cluster-size main effect $\beta_N$ (Equation \ref{app:eq:N.c}), respectively.

The proof for the consistency of the inverse cluster-size weighted estimating equations for the cATE then follows as in the case with the unweighted estimating equations for the iATE.
Accordingly, we can use Equations (\ref{app:eq:baseline.c}-\ref{app:eq:N.c}) to prove the following equalities 
\begin{equation}
\label{app:eq.ee.ZN.c}
    E\left[\left(\frac{1}{1+(N-1)\underline{\rho}}\right) \bm{1}_N^\top \left\{\left[\bm{Y}(1)-\bm{Y}(0)\right] - \left[\underline{\bm{\mu}}(1) - \underline{\bm{\mu}}(0)]\right]\right\}\right] = 0
\end{equation}
and
\begin{equation}
\label{app:eq.ee.Z.c}
    E\left[\left(\frac{1}{N}\right)\left(\frac{1}{1+(N-1)\underline{\rho}}\right) \bm{1}_N^\top \left\{\left[\bm{Y}(1)-\bm{Y}(0)\right] - \left[\underline{\bm{\mu}}(1) - \underline{\bm{\mu}}(0)]\right]\right\}\right] = 0 \,.
\end{equation}
We can then demonstrate that $\Delta_{cATE} = E[(1/N)\bm{1}_N^\top \{\bm{Y}(1) - \bm{Y}(0)\}] = E[(1/N)\bm{1}_N^\top \{\underline{\bm{\mu}}(1) - \underline{\bm{\mu}}(0)\}] = \underline{\Delta}_{CSLMMw-g}$ and complete the proof of consistency. For the previous statement to be true, then $E[(1/N)\bm{1}_N^\top \{[\bm{Y}(1) - \bm{Y}(0) ]- [\underline{\bm{\mu}}(1) - \underline{\bm{\mu}}(0)]\} ] = 0$ needs to also be true.
We prove this is the case with the previously proven results in Equations (\ref{app:eq.ee.ZN.c}) and (\ref{app:eq.ee.Z.c})
\[
\begin{split}
    E&\left[\left(\frac{1}{N}\right)\bm{1}_N^\top \{[\bm{Y}(1) - \bm{Y}(0) ]- [\underline{\bm{\mu}}(1) - \underline{\bm{\mu}}(0)]\}\right] \\
    &= E\left[ \left(\frac{1}{N}\right)\left(\frac{1+(N-1)\underline{\rho}}{1+(N-1)\underline{\rho}}\right) \bm{1}_N^\top \{[\bm{Y}(1) - \bm{Y}(0) ]- [\underline{\bm{\mu}}(1) - \underline{\bm{\mu}}(0)]\}\right] \\
    &\begin{aligned}
            = E&\left[\left(\frac{1}{N}\right) \left(\frac{1}{1+(N-1)\underline{\rho}}\right) \bm{1}_N^\top \{[\bm{Y}(1) - \bm{Y}(0) ]- [\underline{\bm{\mu}}(1) - \underline{\bm{\mu}}(0)]\}\right] \\
            &+ \underline{\rho}\left( E\left[ \left(\frac{1}{1+(N-1)\underline{\rho}}\right) \bm{1}_N^\top \{[\bm{Y}(1) - \bm{Y}(0) ]- [\underline{\bm{\mu}}(1) - \underline{\bm{\mu}}(0)]\}\right] \right) \\
            &- \underline{\rho}\left( E\left[ \left(\frac{1}{N}\right) \left(\frac{1}{1+(N-1)\underline{\rho}}\right) \bm{1}_N^\top \{[\bm{Y}(1) - \bm{Y}(0) ]- [\underline{\bm{\mu}}(1) - \underline{\bm{\mu}}(0)]\}\right] \right)
    \end{aligned}\\
    &=0 \,.
\end{split}
\]
Therefore
\[
\begin{split}
    \underline{\Delta}_{CSLMMw-g} &= E\left[\frac{1}{N} \bm{1}_N^\top \{\underline{\bm{\mu}}(1) - \underline{\bm{\mu}}(0)\}\right] \\
    &= E\left[\frac{1}{N} \bm{1}_N^\top \{\bm{Y}(1) - \bm{Y}(0)\}\right] = \Delta_{cATE} 
\end{split}
\]
completing the proof of consistency. $\square$

\subsection{Proof of Lemma \ref{lemma:g}.}
It suffices to show $\bm{R}_i^{-1}\bm{\mathcal{Z}}_i^{-1/2} = \bm{\mathcal{Z}}_i^{-1/2}\bm{R}_i^{-1}$.
For supplementary assumption (I), with $\rho=0$, we have a working independence correlation $\bm{R}_i = \textbf{I}_{N_i}$, which directly implies the desired result.
For supplementary assumption (II) that $v(Y_{ik}) \equiv \sigma^2$, then $\bm{\mathcal{Z}}_i=\sigma^2 \textbf{I}_{N_i}$ is a diagonal matrix, which yields the desired result.
For supplementary assumption (III), we have $\bm{\mathcal{Z}}_i = v_{i}I_{N_i}$ for a variance function $v_{i}$ common across subjects within each cluster $i$ since $v(Y_{ik})$ is a function of $Z_i$, $N_i$, and $\bm{X_{ik}}$, which are all constant across $k$.
Then $\bm{R}_i^{-1} \bm{\mathcal{Z}}_i^{-1/2} = \bm{\mathcal{Z}}_i^{-1/2} \bm{R}_i^{-1}$, which completes the proof.
$\square$

\subsection{Proof of Theorem \ref{Theorem_GEE_g}.a.}

Consider the estimating equations $\sum_{i=1}^m \bm{\psi}(\bm{O}_i, \bm{\theta})$, where
\[
\begin{split}
    \bm{\psi}(\bm{O}_i, \bm{\theta})
    & = 
        \left(
        \begin{gathered}
            \left(\frac{\sum_{s=1}^{m}N_s}{m}\right) \Delta_{CSGEE-g} - \bm{1}_{N_i}^\top(\bm{\mu}_i(1) - \bm{\mu}_i(0)) \\
            \bm{Q}_i^\top \bm{R}_i^{-1} (\bm{Y}_i - \bm{\mu}_i) \\
            f(\bm{O}_i, \rho)
        \end{gathered}
        \right)
\end{split}
\]
where
$\bm{\mu}_i = g^{-1}(\bm{Q}_i\bm{\beta})$ and
$\bm{\mu}_i(a) = g^{-1}\left(\bm{Q}_i(a) \bm{\beta}\right) = \left(g^{-1}\left(\beta_{0}  + a \beta_Z + a N_i  \beta_{ZN} + N_i \beta_{N} + \bm{\beta}_X^\top \bm{X}_{ik}\right)\right)_{k=1,...,N_i} \in \mathbb{R}^{N_i}$, $a\in\{0,1\}$.
With a canonical link function $U_i=\frac{d\bm{\mu}_i}{d\bm{\beta}}=\bm{\mathcal{Z}}_i\bm{Q}_i$; then with Lemma \ref{lemma:g}, $\bm{U}_i^\top \bm{\mathcal{Z}}_i^{-1/2} \bm{R}_i^{-1} \bm{\mathcal{Z}}_i^{-1/2} (\bm{Y}_i - \bm{\mu}_i)$ in Equation (\ref{eq:GEE}) simplifies to $\bm{Q}_i^\top \bm{R}_i^{-1} (\bm{Y}_i - \bm{\mu}_i)$.
The estimator for $\bm{\theta}=(\Delta_{CSGEE-g}, \bm{\beta}, \rho)^{\top} \in \mathbb{R}^{6+p}$ is then defined as a solution to the estimating equation
\[
    \sum_{i=1}^{m} \bm{\psi}(\bm{O}_i;\bm{\theta})=0 \,.
\]
We hence define the estimating function as
\begin{equation}
    \bm{\psi}(\bm{O};\bm{\theta}) =  \left(
        \begin{gathered}
            E[N] \Delta_{CSGEE-g} - \bm{1}_{N}^\top(\bm{\mu}(1) - \bm{\mu}(0))  \\
            \bm{Q}^\top \bm{R}^{-1} (\bm{Y} - \bm{\mu}) \\
            f(\bm{O}, \rho)
        \end{gathered}
        \right)
\end{equation}
Recall that subscript $i$ is omitted when taking the expectation with respect to distribution $\mathcal{P}$.

The proof in this generalized setting then extends exactly as in the linear setting, leading again to the following equalities
\begin{equation}
    E\left[N \left(\frac{1}{1+(N-1)\underline{\rho}}\right) \bm{1}_N^\top \left\{\left[\bm{Y}(1)-\bm{Y}(0)\right] - \left[\underline{\bm{\mu}}(1) - \underline{\bm{\mu}}(0)]\right]\right\}\right] = 0
\end{equation}
and
\begin{equation}
    E\left[\left(\frac{1}{1+(N-1)\underline{\rho}}\right) \bm{1}_N^\top \left\{\left[\bm{Y}(1)-\bm{Y}(0)\right] - \left[\underline{\bm{\mu}}(1) - \underline{\bm{\mu}}(0)]\right]\right\}\right] = 0 \,.
\end{equation}
We can then demonstrate that implementing CSGEE-g (g-computation with the cluster-size saturated working GEE) results in
\[
\begin{split}
   \underline{\Delta}_{CSGEE-g} &= \frac{E[\bm{1}_N^\top \{\underline{\bm{\mu}}(1) - \underline{\bm{\mu}}(0)\}]}{E[N]} \\
    &= \frac{E[\bm{1}_N^\top \{\bm{Y}(1) - \bm{Y}(0)\}]}{E[N]} = \Delta_{iATE} 
\end{split}
\]
completing the proof of consistency. $\square$

\subsection{Proof of Theorem \ref{Theorem_GEE_g}.b.}

Consider the estimating equations $\sum_{i=1}^m \bm{\psi}(\bm{O}_i, \bm{\theta})$, where
\[
\begin{split}
    \bm{\psi}(\bm{O}_i, \bm{\theta})
    & = 
        \left(
        \begin{gathered}
            \Delta_{CSGEEw-g} - \frac{1}{N_i}\bm{1}_{N_i}^\top(\bm{\mu}_i(1) - \bm{\mu}_i(0)) \\
            \frac{1}{N_i}\bm{Q}_i^\top \bm{R}_i^{-1} (\bm{Y}_i - \bm{\mu}_i) \\
            f(\bm{O}_i, \rho)
        \end{gathered}
        \right)
\end{split}
\]
The estimator for $\bm{\theta}=(\Delta_{CSGEEw-g}, \bm{\beta}, \rho)^{\top} \in \mathbb{R}^{6+p}$ is then defined as a solution to the estimating equation
\[
    \sum_{i=1}^{m} \bm{\psi}(\bm{O}_i;\bm{\theta})=0 \,.
\]
We hence define the estimating function as
\begin{equation}
    \bm{\psi}(\bm{O};\bm{\theta}) =  \left(
        \begin{gathered}
            \Delta_{CSGEEw-g} - \frac{1}{N}\bm{1}_{N}^\top(\bm{\mu}(1) - \bm{\mu}(0)) \\
            \frac{1}{N} \bm{Q}^\top \bm{R}^{-1} (\bm{Y} - \bm{\mu}) \\
            f(\bm{O}, \rho)
        \end{gathered}
        \right)
\end{equation}
Again, the proof in this generalized setting then extends exactly as in the linear setting, leading again to the following equalities
\begin{equation}
    E\left[\left(\frac{1}{1+(N-1)\underline{\rho}}\right) \bm{1}_N^\top \left\{\left[\bm{Y}(1)-\bm{Y}(0)\right] - \left[\underline{\bm{\mu}}(1) - \underline{\bm{\mu}}(0)]\right]\right\}\right] = 0
\end{equation}
and
\begin{equation}
    E\left[\left(\frac{1}{N}\right)\left(\frac{1}{1+(N-1)\underline{\rho}}\right) \bm{1}_N^\top \left\{\left[\bm{Y}(1)-\bm{Y}(0)\right] - \left[\underline{\bm{\mu}}(1) - \underline{\bm{\mu}}(0)]\right]\right\}\right] = 0 \,.
\end{equation}
We can then demonstrate that implementing g-computation with the cluster-size saturated working GEE with inverse cluster-size weights results in
\[
\begin{split}
    \underline{\Delta}_{CSGEEw-g} &= E\left[\frac{1}{N} \bm{1}_N^\top \{\underline{\bm{\mu}}(1) - \underline{\bm{\mu}}(0)\}\right] \\
    &= E\left[\frac{1}{N} \bm{1}_N^\top \{\bm{Y}(1) - \bm{Y}(0)\}\right] = \Delta_{cATE} 
\end{split}
\]
completing the proof of consistency. $\square$

\subsection{Proof of Theorem \ref{Theorem_g=MRS}}
We can generally prove Theorem \ref{Theorem_g=MRS} by demonstrating that
\[
\sum_{i=1}^{m} \frac{\lambda_i}{\sum_{s=1}^m\lambda_s}
\Biggl\{
\underbrace{\frac{I\{Z_i=a\}\bigl(\bar Y_i-\widehat{\bar\mu}_i(a)\bigr)}
{\pi^{a}(1-\pi)^{1-a}}}_{\text{weighted cluster-level residual}}
\Biggr\} = 0
\]
for each $a \in \{0,1\}$, when cluster-saturation is specified within the g-computation or model-robust standardization.
The above expression can be algebraically reorganized into
\begin{equation}
\label{eq:g=MRS}
\sum_{i=1}^{m} \frac{\lambda_i}{N_i} I\{Z_i=a\} \bm{1}_{N_i}^\top\left(\bm{Y}_i-\hat{\bm{\mu}}_i(a)\right) = 0
\end{equation}
which we will prove to be the case.

As described in the proofs for Theorems \ref{Theorem_g} \& \ref{Theorem_GEE_g}, under the described assumptions and regularity conditions, recall the estimating equations 
$\sum_{i=1}^m \bm{\psi}(\bm{O}_i, \bm{\theta})=0$
where
\[
\begin{split}
    \bm{\psi}(\bm{O}_i, \bm{\theta})
    & = 
        \left(
        \begin{gathered}
             \left(\frac{\sum_{s=1}^{m}\lambda_s}{m}\right) \Delta_{CSGEE\bm{\lambda}-g} -  \frac{\lambda_i}{N_i}\bm{1}_{N_i}^\top(\bm{\mu}_i(1) - \bm{\mu}_i(0)) \\
            \frac{\lambda_i}{N_i}\bm{Q}_i^\top \bm{R}_i^{-1} (\bm{Y}_i - \bm{\mu}_i) \\
            f(\bm{O}_i, \rho)
        \end{gathered}
        \right)
\end{split} =0
\]
with $\lambda_i=N_i$ and $\lambda_i=1$ yielding the individual-level and cluster-level estimators, respectively.
Then the second through fifth entries of $\sum_{i=1}^m \bm{\psi}(\bm{O}_i, \hat{\bm{\theta}})=0$ are
\begin{gather}
\label{eq:baseline.c.general}
    \sum_{i=1}^m \frac{\lambda_i}{N_i}\bm{1}_{N_i}^\top \bm{\hat{R}}_i^{-1}(\bm{Y}_i-\hat{\bm{\mu}}_i) = 0 \\
\label{eq:treatment.c.general}
    \sum_{i=1}^m \frac{\lambda_i}{N_i} I\{Z_i=1\} \bm{1}_{N_i}^\top \bm{\hat{R}}_i^{-1}(\bm{Y}_i-\hat{\bm{\mu}}_i) = 0 \\
\label{eq:treatment_N.c.general}
    \sum_{i=1}^m \lambda_i I\{Z_i=1\} \bm{1}_{N_i}^\top \bm{\hat{R}}_i^{-1}(\bm{Y}_i-\hat{\bm{\mu}}_i) = 0 \\
\label{eq:N.c.general}
    \sum_{i=1}^m \lambda_i \bm{1}_{N_i}^\top \bm{\hat{R}}_i^{-1}(\bm{Y}_i-\hat{\bm{\mu}}_i) = 0
\end{gather}
corresponding to the score functions for the model intercept $\beta_0$, treatment effect coefficient $\beta_Z$, treatment $\times$ cluster-size interaction $\beta_{ZN}$, and cluster-size main effect $\beta_N$, respectively.
It is then trivial to algebraically demonstrate that the above estimating equations are respectively equal to
\[
\begin{split}
\sum_{i=1}^m \frac{\lambda_i}{N_i}  \left(\frac{1}{1+(N_i-1)\hat{\rho}}\right) \bm{1}_{N_i}^\top (\bm{Y}_i-\hat{\bm{\mu}}_i) &= 0 \,, \\
\sum_{i=1}^m \frac{\lambda_i}{N_i} I\{Z_i=1\} \left(\frac{1}{1+(N_i-1)\hat{\rho}}\right) \bm{1}_{N_i}^\top (\bm{Y}_i-\hat{\bm{\mu}}_i) &= 0 \,, \\
\sum_{i=1}^m \lambda_i I\{Z_i=1\} \left(\frac{1}{1+(N_i-1)\hat{\rho}}\right)\bm{1}_{N_i}^\top (\bm{Y}_i-\hat{\bm{\mu}}_i) &= 0 \,, \\
\sum_{i=1}^m \lambda_i \left(\frac{1}{1+(N_i-1)\hat{\rho}}\right) \bm{1}_{N_i}^\top (\bm{Y}_i-\hat{\bm{\mu}}_i) &= 0 \,.
\end{split}
\]
Taking the difference between Equations (\ref{eq:baseline.c.general}) and (\ref{eq:treatment.c.general}), as well as between Equations (\ref{eq:N.c.general}) and (\ref{eq:treatment_N.c.general}), subsequently yield
\[
\begin{split}
\sum_{i=1}^m \frac{\lambda_i}{N_i} I\{Z_i=0\} \left(\frac{1}{1+(N_i-1)\hat{\rho}}\right) \bm{1}_{N_i}^\top (\bm{Y}_i-\hat{\bm{\mu}}_i) &= 0 \,, \\
\sum_{i=1}^m \lambda_i I\{Z_i=0\} \left(\frac{1}{1+(N_i-1)\hat{\rho}}\right) \bm{1}_{N_i}^\top (\bm{Y}_i-\hat{\bm{\mu}}_i) &= 0 \,, \\
\end{split}
\]

Altogether, we can then prove the equality in Equation (\ref{eq:g=MRS})
\[
\begin{split}
    \sum_{i=1}^{m} & \frac{\lambda_i}{N_i} I\{Z_i=a\} \bm{1}_{N_i}^\top\left(\bm{Y}_i-\hat{\bm{\mu}}_i(a)\right)  \\
    &= \sum_{i=1}^{m} \frac{\lambda_i}{N_i} I\{Z_i=a\} \left(\frac{1+(N_i-1)\hat{\rho}}{1+(N_i-1)\hat{\rho}}\right)\bm{1}_{N_i}^\top\left(\bm{Y}_i-\hat{\bm{\mu}}_i(a)\right) \\
    &\begin{aligned}
        = \sum_{i=1}^{m} & \frac{\lambda_i}{N_i} I\{Z_i=a\} \left(\frac{1}{1+(N_i-1)\hat{\rho}}\right)\bm{1}_{N_i}^\top\left(\bm{Y}_i-\hat{\bm{\mu}}_i(a)\right) \\
        & + \hat{\rho} \sum_{i=1}^{m} \lambda_i I\{Z_i=a\} \left(\frac{1}{1+(N_i-1)\hat{\rho}}\right)\bm{1}_{N_i}^\top\left(\bm{Y}_i-\hat{\bm{\mu}}_i(a)\right) \\
        & - \hat{\rho} \sum_{i=1}^{m} \frac{\lambda_i}{N_i} I\{Z_i=a\} \left(\frac{1}{1+(N_i-1)\hat{\rho}}\right)\bm{1}_{N_i}^\top\left(\bm{Y}_i-\hat{\bm{\mu}}_i(a)\right) \\
    \end{aligned} \\
    &= 0
\end{split}
\]
as per the equalities derived above for $a \in \{0,1\}$. Altogether, this completes the proof for both the cluster-size saturated individual-level ($\lambda_i=N_i$; $\hat{\Delta}_{CSGEE-g} = \hat{\Delta}_{CSGEE-MRS}$) and cluster-level ($\lambda_i=1$; $\hat{\Delta}_{CSGEEw-g} = \hat{\Delta}_{CSGEEw-MRS}$) estimators.
$\square$

\section{Extended simulation results}
\label{app:extended_sim}

As per Li et al. \cite{li_model-robust_2025}, scenarios with continuous outcomes, ICS.I, ICS.II, and $m = 30$ clusters are simulated as
\[
    Y_{ik} = \frac{H_{1i}X_{1ik}^{2}}{5 N_i}
    + \cos(H_{2i})\, X_{2ik}
    -\frac{N_i^{2} \log(N_i)}{(E[N_i])^{2}}
    +  |H_{2i}| \sin(X_{2ik})
    +  I\{Z_i=1\}\frac{N_i^{2} \log(N_i)}{(E[N_i])^{2}}
    + I\{Z_i=1\} \alpha_i + \epsilon_{ik}, 
\]
Cluster sizes are drawn as $N_i \sim Uniform\{20,  180\}$ with expected cluster size $E[N_i] = 100$.
A binary cluster-level covariate $H_{1i} \sim Bernoulli\left( \Phi(\sin (N_i)) \right)$ (where $\Phi(\cdot)$ denotes the standard normal cumulative distribution function), a continuous cluster-level covariate $H_{2i} \sim N\left( 2 + H_{1i} N_i / 10,\, 9 \right)$, a continuous individual-level covariate $X_{1ik} \sim N\left( H_{1i} H_{2i} + N_i / 100,\, 16 \right)$, and a binary individual-level covariate $X_{2ik} \sim Bernoulli\left( expit\left[\log(N_i)\, X_{1ik}\, H_{1i} + H_{2i} \right] \right)$ are included.
Finally, cluster random intercepts and residuals are generated as $\alpha_i \sim N(0, \tau^2 = 0.2)$ and $\epsilon_{ik} \sim N(0, 1)$.
This data-generating process targets $iATE = 8.15$ and $cATE = 5.92$, as reported in Li et al. \cite{li_model-robust_2025}.

As per Li et al. \cite{li_model-robust_2025}, scenarios with binary outcomes, ICS.I, ICS.II, and $m = 30$ clusters are simulated as
\[
logit(E[Y_{ik}|\alpha_i]) =  
              -\frac{N_i^{2} \log(N_i)}{5 (E[N_i])^{2}}
             + \frac{X_{1ik}^{2}}{2 N_i}
             + H_{1i}
             + \cos(H_{2i})\, X_{2ik}
             + \frac{|H_{2i}|}{5} + I\{Z_i=1\}\frac{N_i^{2} \log(N_i)}{5 (E[N_i])^{2}}
              + I\{Z_i=1\}\alpha_i ,
\]
Cluster sizes are drawn as $N_i \sim Uniform\{20, 180\}$ with expected cluster size $E[N_i] = 100$.
A binary cluster-level covariate $H_{1i} \sim Bernoulli\left(0.5\right)$, a continuous cluster-level covariate $H_{2i} \sim N\left(2 + H_{1i} + N_i / E[N_i],\, 1 \right)$, a continuous individual-level covariate
$X_{1ik} \sim N\left( H_{1i} + H_{2i} / 20 + N_i / 100,\, 16 \right)$, and a binary individual-level covariate $X_{2ik} \sim Bernoulli\left( expit\left[\log(N_i)\, X_{1ik}\, H_{1i} + H_{2i} \right] \right)$ are included.
Finally, cluster random intercepts are generated as $\alpha_i \sim N(0, \tau^2 = 0.2)$.
This data-generating process targets marginal log-odds-ratio estimands of $iALOR = 1.24$ and $cALOR = 0.91$, as reported in Li et al. \cite{li_model-robust_2025}.

Complete and extended simulation results are included below as tables. This includes results from the appropriately weighted cluster-size saturated independence estimating equation with g-computation (CSIEE-g, CSIEEw-g) and the corresponding cluster main-effect models (CMLMM, CMLMMw, CMGEE, CMGEEw) in both simulation scenarios 1 and 2, and the generalized linear mixed-effects model with model-robust standardization (GLMM-MRS, GLMMw-MRS) in simulation scenario 2.

Furthermore, we report the simulation results from these models in the simulation replicates of Li et al. \cite{li_model-robust_2025}.
The described simulation scenarios from Li et al. \cite{li_model-robust_2025} include cluster-level and individual-level covariates; accordingly we additionally explore CSGEE-g and CSGEEw-g with adjustment for cluster-level covariates (``clus cov adj'') or adjustment for both cluster and individual-level covariates (``all cov adj'').

\subsection{Extended simulation results from scenario 1 (continuous outcomes)}

\begin{table}[H]
\centering
\caption{Analysis Results targeting iATE and cATE, with $m=40$ clusters across $1000$ simulation replicates}
\begin{tabular}{llcccc}
\toprule
Analysis & Estimand & Rel Bias (\%) & Efficiency & Avg Var & CP \\
\midrule
IEE      & iATE & -0.263 & 0.298 & 0.321 & 0.958 \\
LMM      & iATE & -22.886 & 0.483 & 0.519 & 0.863 \\
CMLMM & iATE & -22.193 & 0.190 & 0.183 & 0.655 \\
LMM-MRS  & iATE & -0.491 & 0.424 & 0.588 & 0.973 \\
CSIEE-g  & iATE & 0.755  & 0.176 & 0.176 & 0.940 \\
CSLMM-g  & iATE & 0.459  & 0.180 & 0.175 & 0.936 \\
\midrule
IEEw     & cATE & 0.480  & 0.484 & 0.520 & 0.946 \\
LMMw     & cATE & -40.351 & 0.591 & 0.629 & 0.718 \\
CMLMMw & cATE & -39.462 & 0.231 & 0.254 & 0.540 \\
LMMw-MRS & cATE & 0.480  & 0.484 & 0.534 & 0.947 \\
CSIEEw-g & cATE & 0.743  & 0.186 & 0.173 & 0.937 \\
CSLMMw-g & cATE & 0.437  & 0.185 & 0.174 & 0.939 \\
\bottomrule
\end{tabular}
\end{table}

\subsection{Extended simulation results from scenario 2 (binary outcomes)}

\begin{table}[H]
\centering
\caption{Analysis Results targeting iAOR and cAOR, with $m=40$ clusters across $1000$ simulation replicates}
\begin{tabular}{llcccc}
\toprule
Analysis & Estimand & Rel Bias (\%) & Efficiency & Avg Var & CP \\
\midrule
IEE       & iAOR &   1.466 &    0.379 &  0.400   & 0.950 \\
GLMM      & iAOR &  -3.240 &    0.510 &  0.521   & 0.926 \\
GEE       & iAOR & -14.338 &    0.365 &  0.387   & 0.814 \\
CMGEE & iAOR & -10.228 & 0.247 & 0.291 & 0.866 \\
GLMM-MRS  & iAOR &  22.814 & 8255.551 & $\infty$ & 0.964 \\
GEE-MRS   & iAOR &   1.619 &    0.444 &  0.503   & 0.953 \\
CSIEE-g   & iAOR &   1.437 &    0.281 &  0.341   & 0.961 \\
CSGLMM-g  & iAOR &   3.181 &    0.300 &  0.362   & 0.966 \\
CSGEE-g   & iAOR &   1.363 &    0.279 &  0.328   & 0.959 \\
\midrule
IEEw      & cAOR &   3.853 &    0.327 &  0.348   & 0.949 \\
GLMMw     & cAOR &  -9.203 &    0.438 &  0.495   & 0.863 \\
GEEw      & cAOR & -19.377 &    0.313 &  0.364   & 0.737 \\
CMGEEw & cAOR & -13.555 & 0.237 & 0.282 & 0.854 \\
GLMMw-MRS & cAOR &   3.853 &    0.327 & 11.483   & 1.000 \\
GEEw-MRS  & cAOR &   3.853 &    0.327 &  0.367   & 0.951 \\
CSIEEw-g  & cAOR &   2.127 & 0.174 & 0.203 & 0.965 \\
CSGLMMw-g & cAOR &   3.119 &    0.185 &  0.218   & 0.964 \\
CSGEEw-g  & cAOR &   2.066 & 0.174 & 0.204 & 0.962 \\
\bottomrule
\end{tabular}
\end{table}

\begin{table}[H]
\caption{
    Power and average values for the test of informative cluster size between the unweighted and inverse cluster-size weighted estimators.
    Results are reported for simulation scenarios 1 \& 2 with $m=40$ clusters.
}
\begin{center}
\bgroup
\def\arraystretch{1.3}
{
\begin{tabular}{|c c c|} 
    \hline
    \textbf{Estimator} & \textbf{Average Test Statistic} & \textbf{Power} \\
    \hline\hline
    Scenario 1 (continuous) \\
    \hline\hline
    IEE vs IEEw & -1.981 & 0.483\\
    \hdashline
    LMM-MRS vs LMMw-MRS & -1.592 & 0.245\\
    \hdashline
    CSIEE-g vs CSIEEw-g & -3.631 & 0.984 \\
    \hdashline
    CSLMM-g vs CSLMMw-g & -3.572 & 0.985\\
    \hline\hline
    Scenario 2 (binary) \\
    \hline\hline 
    IEE vs IEEw & -2.200 & 0.619\\
    \hdashline
    GEE-MRS vs GEEw-MRS & -1.790 & 0.375 \\
    \hdashline
    GLMM-MRS vs GLMMw-MRS & 0.081 & 0.001 \\
    \hdashline
    CSIEE-g vs CSIEEw-g & -2.438 & 0.727 \\
    \hdashline
    CSGLMM-g vs CSGLMMw-g & -2.444 & 0.727 \\
    \hdashline
    CSGEE-g vs CSGEEw-g & -2.433 & 0.728 \\
    \hline
\end{tabular}
}
\egroup
\end{center}
\end{table}

\subsection{Simulation replicates of Li et al. \cite{li_model-robust_2025}}

\begin{table}[H]
\centering
\caption{Analysis Results for iATE and cATE, with $m=30$ clusters across $1000$ simulation replicates}
\begin{tabular}{llcccc}
\toprule
Analysis & Estimand & Relative Bias (\%) & Efficiency & Avg Var & Coverage Probability \\
\midrule
IEE       & iATE &  -1.673 & 6.234 & 6.220 & 0.923 \\
LMM       & iATE & -27.203 & 3.814 & 3.814 & 0.760 \\
CMLMM & iATE & -27.439 & 3.763 & 4.225 & 0.783 \\
LMM-MRS   & iATE &  -1.829 & 6.284 & 6.376 & 0.925 \\
CSIEE-g   & iATE &  -0.784 & 6.083 & 6.764 & 0.929 \\
CSLMM-g   & iATE &  -1.002 & 6.045 & 6.518 & 0.934 \\
\midrule
IEEw      & cATE &   0.048 & 3.815 & 3.812 & 0.938 \\
LMMw      & cATE & -38.252 & 2.644 & 2.855 & 0.703 \\
CMLMMw & cATE & -37.293 & 2.771 & 3.355 & 0.743 \\
LMMw-MRS  & cATE &   0.048 & 3.815 & 3.944 & 0.942 \\
CSIEEw-g  & cATE &  -0.229 & 3.642 & 3.997 & 0.952 \\
CSLMMw-g  & cATE &  -0.396 & 3.656 & 3.947 & 0.948 \\
\bottomrule
\end{tabular}
\end{table}

\begin{table}[H]
\centering
\caption{Analysis Results for the individual-average and cluster-average log-odds ratio (iALOR, cALOR), with $m=30$ clusters across $1000$ simulation replicates}
\label{app:tab:lietal_bin}
\begin{tabular}{llcccc}
\toprule
Analysis & Estimand & Relative Bias (log; \%) & Efficiency & Avg Var & Coverage Probability \\
\midrule
IEE       & iALOR &   0.953 & 2.509 &     2.863 & 0.932 \\
GLMM      & iALOR &  -3.708 & 3.023 &     3.407 & 0.905 \\
GEE       & iALOR & -24.618 & 1.045 &     1.147 & 0.741 \\
CMGEE & iALOR & -21.796 & 1.186 & 1.41 & 0.77 \\
GLMM-MRS  & iALOR &   5.172 & 9.429 & 12235.265 & 0.989 \\
GEE-MRS   & iALOR &   0.777 & 2.587 &     3.061 & 0.940 \\
CSIEE-g   & iALOR &   1.776 & 2.425 &     2.933 & 0.929 \\
CSGLMM-g  & iALOR &   4.280 & 2.695 &     3.228 & 0.934 \\
CSGEE-g   & iALOR &   1.676 & 2.405 &     2.616 & 0.930 \\
CSGEE-g (clus cov adj) & iALOR & 0.991 & 0.874 & 1.076 & 0.947 \\
CSGEE-g (all cov adj) & iALOR & 0.848 & 0.862 & 1.062 & 0.943 \\
\midrule
IEEw      & cALOR &   0.748 & 0.996 &     1.091 & 0.944 \\
GLMMw     & cALOR & -21.239 & 1.370 &     1.615 & 0.857 \\
GEEw      & cALOR & -36.797 & 0.623 &     0.718 & 0.741 \\
CMGEEw & cALOR & -32.749 & 0.672 & 0.9 & 0.788 \\
GLMMw-MRS & cALOR &   0.748 & 0.996 &     9.551 & 0.999 \\
GEEw-MRS  & cALOR &   0.748 & 0.996 &     1.129 & 0.944 \\
CSIEEw-g  & cALOR &  0.779 & 0.895 & 1.615 & 0.938 \\
CSGLMMw-g & cALOR &   3.198 & 0.967 &     1.120 & 0.944 \\
CSGEEw-g  & cALOR & 0.71 & 0.892 & 1.023 & 0.938 \\
CSGEEw-g (clus cov adj) & cALOR & 1.096 & 0.414 & 0.506 & 0.966 \\
CSGEEw-g (all cov adj) & cALOR & 1.168 & 0.409 & 0.502 & 0.966 \\
\bottomrule
\end{tabular}
\end{table}

Adjusting for cluster-level covariates in CSGEE-g and CSGEEw-g (``clus cov adj'') yielded consistent estimators (Theorem \ref{Theorem_GEE_g}) that were unbiased and had notable improvements in efficiency (Table \ref{app:tab:lietal_bin}), despite misspecification of the cluster-level covariate structure.

While the described CSGEE-g and CSGEEw-g with adjustment for both cluster and individual-level covariates (``all cov adj'') is not consistent (Lemma \ref{lemma:g}, Theorem \ref{Theorem_GEE_g}), we observe that adjusting for individual-level covariates can still return minimally biased results.
This is despite violating all supplementary assumptions in Lemma \ref{lemma:g}, and particularly supplementary assumption (III).
These minimally biased results can be characterized by first starting with the estimating equation in Equation (\ref{eq:GEE}) with $E[\bm{\psi}(\bm{O};\bm{\theta})]=0$ in the absence of Lemma 2's supplementary assumptions (I)-(III)
\[
\begin{split}
    E &\left[ \underline{\bm{U}}^\top \underline{\bm{\mathcal{Z}}}^{-1/2} \underline{\bm{R}}^{-1} \underline{\bm{\mathcal{Z}}}^{-1/2} (\bm{Y} - \underline{\bm{\mu}}) \right]\\
    &= E\left[ \underline{\bm{Q}}^\top \underline{\bm{\mathcal{Z}}} \, \underline{\bm{\mathcal{Z}}}^{-1/2} \underline{\bm{R}}^{-1} \underline{\bm{\mathcal{Z}}}^{-1/2} (\bm{Y} - \underline{\bm{\mu}}) \right]\\
    &= E \left[ \underline{\bm{Q}}^\top \underline{\bm{\mathcal{Z}}}^{1/2} \underline{\bm{R}}^{-1} \underline{\bm{\mathcal{Z}}}^{-1/2} (\bm{Y} - \underline{\bm{\mu}}) \right]\\
    &=0 \,.
\end{split}
\]
Then with the second through fifth entries of the estimating equations that are used in the proof for Theorem \ref{Theorem_GEE_g}, corresponding to the score functions for the model intercept $\beta_0$, treatment effect coefficient $\beta_Z$, treatment $\times$ cluster-size interaction $\beta_{ZN}$, and cluster-size main effect $\beta_N$, are then
\begin{gather}
    E\left[ \frac{\lambda}{N}\bm{1}_{N}^\top \underline{\bm{\mathcal{Z}}}^{1/2} \underline{\bm{R}}^{-1} \underline{\bm{\mathcal{Z}}}^{-1/2}(\bm{Y}-\underline{\bm{\mu}}) \right] =  E\left[ \frac{\lambda}{N} \underline{\bm{\Omega}}^\top (\bm{Y}-\underline{\bm{\mu}}) \right]  = 0 \\
     E\left[ \frac{\lambda}{N} I\{Z=1\} \bm{1}_{N}^\top \underline{\bm{\mathcal{Z}}}^{1/2} \underline{\bm{R}}^{-1} \underline{\bm{\mathcal{Z}}}^{-1/2} (\bm{Y}-\underline{\bm{\mu}}) \right] =  E\left[ \frac{\lambda}{N} I\{Z=1\} \underline{\bm{\Omega}}^\top (\bm{Y}-\underline{\bm{\mu}}) \right] = 0 \\
     E\left[  \lambda I\{Z=1\} \bm{1}_{N}^\top \underline{\bm{\mathcal{Z}}}^{1/2} \underline{\bm{R}}^{-1} \underline{\bm{\mathcal{Z}}}^{-1/2} (\bm{Y}-\underline{\bm{\mu}}) \right]  = E\left[  \lambda I\{Z=1\} \underline{\bm{\Omega}}^\top (\bm{Y}-\underline{\bm{\mu}}) \right]  = 0 \\
     E\left[  \lambda \bm{1}_{N}^\top \underline{\bm{\mathcal{Z}}}^{1/2} \underline{\bm{R}}^{-1} \underline{\bm{\mathcal{Z}}}^{-1/2} (\bm{Y} -\underline{\bm{\mu}}) \right] = E\left[  \lambda \underline{\bm{\Omega}}^\top (\bm{Y} -\underline{\bm{\mu}}) \right]  = 0
\end{gather}
with $\lambda=N$ or $1$ yielding the individual-level or cluster-level estimators, respectively, as per Equations (\ref{eq:baseline.c.general}) - (\ref{eq:N.c.general}).
In the above, $\underline{\bm{\Omega}}^\top=\bm{1}_{N}^\top \underline{\bm{\mathcal{Z}}}^{1/2} \underline{\bm{R}}^{-1} \underline{\bm{\mathcal{Z}}}^{-1/2} \in \mathbb{R}^{1\times N}$ with entries
\[
\begin{split}
    \Omega_{.k} &= \frac{1}{1-\underline{\rho}} \left( 1 - \frac{\underline{\rho} \sum_{l=1}^{N} \sqrt{v(\underline{\mu}_{.l})}}{(1+(N-1)\underline{\rho})\sqrt{v(\underline{\mu}_{.k})}} \right) \\
    &= \frac{1}{1+(N-1)\underline{\rho}} + \left( \frac{\underline{\rho}N}{(1-\underline{\rho})(1+(N-1)\underline{\rho})} \right) \left( 1- \frac{\left(\sum_{l=1}^{N}\sqrt{v(\underline{\mu}_{.l})}\right) /N}{\sqrt{v(\underline{\mu}_{.k})}} \right) \\
    &= \frac{1}{1+(N-1)\underline{\rho}} + \underline{\delta}_{.k} \,.
\end{split}
\]
With any of the supplementary assumptions (I)-(III) detailed in Lemma \ref{lemma:g}, $\underline{\delta}_{.k}=0$ in the above equation, so $\underline{\bm{\Omega}}^\top$ reduces to the uniform per-cluster weight $\{1+(N-1)\underline{\rho}\}^{-1}\bm{1}_{N}^\top$ and the saturating equations collapse to the scalar form used to establish the Theorem \ref{Theorem_GEE_g} result.

When $\underline{\delta}_{.k}\neq 0$, two consequences arise. First, the saturating equations weight individuals non-uniformly through $\underline{\bm{\Omega}}$ rather than through the uniform scalar, distorting the treatment contrast. Second, because $\sqrt{v(\underline{\mu}_{.k})}$ is evaluated at the fitted mean under the observed arm, $\underline{\delta}_{.k}$ is itself treatment arm-dependent; consequently $\bm{1}_{N}^\top(\bm{Y}(0)-\underline{\bm{\mu}}(0))$, which vanishes under cluster randomization (A2) in the proof of Theorem \ref{Theorem_GEE_g}, no longer cancels on its own and contributes an additional bias term. Both channels are governed by the same product of $\underline{\rho}$ and the within-cluster spread of $\sqrt{v(\underline{\mu}_{.k})}$, and each vanishes under conditions (I)-(III).

In the described replicates of Li et al.'s binary simulation results \cite{li_model-robust_2025}, where we violate supplementary assumptions (I)-(III) detailed in Lemma \ref{lemma:g}, the bias from $\underline{\delta}_{.k}$ is minimized due to $\tau^2 = 0.2$ on the logit scale keeping the marginal working $\rho$ low.
Furthermore, the individual-level covariate induced fluctuations in $\sqrt{v(\underline{\mu}_{.k})}=\sqrt{\underline{\mu}_{.k}(1-\underline{\mu}_{.k})}$ varies little for $\underline{\mu}_{.k}$ near $1/2$, where this data-generating process concentrates; even sizable within-cluster heterogeneity in the fitted probabilities therefore induces negligible spread in $\sqrt{v(\underline{\mu}_{.k})}$.
Altogether, $\underline{\delta}_{.k} \approx 0$ and $\Omega_{.k} \approx \frac{1}{1+(N-1)\underline{\rho}}$, leading to the observed minimally biased results.

\begin{table}[H]
\caption{
    Power and average values for the test of informative cluster size between the unweighted and inverse cluster-size weighted estimators.
    Results are reported for replications of  Li et al.'s \cite{li_model-robust_2025} simulation scenarios with $m=30$ clusters.
}
\label{tab:ICS_test_lietal}
\begin{center}
\bgroup
\def\arraystretch{1.3}
{
\begin{tabular}{|c c c|} 
    \hline
    \textbf{Estimator} & \textbf{Average Test Statistic} & \textbf{Power} \\
    \hline\hline
    Li et al. Scenario (continuous) \\
    \hline\hline
    IEE vs IEEw & -2.238 & 0.574 \\
    \hdashline
    LMM-MRS vs LMMw-MRS & -2.193 & 0.571 \\
    \hdashline
    CSIEE-g vs CSIEEw-g & -2.296 & 0.557 \\
    \hdashline
    CSLMM-g vs CSLMMw-g & -2.311 & 0.571\\
    \hline\hline
    Li et al. Scenario (binary) \\
    \hline\hline 
    IEE vs IEEw & -1.337 & 0.036 \\
    \hdashline
    GLMM-MRS vs GLMMw-MRS & -0.342 & 0.000 \\
    \hdashline
    GEE-MRS vs GEEw-MRS & -1.280 & 0.020 \\
    \hdashline
    CSIEE-g vs CSIEEw-g & -1.356 & 0.055 \\
    \hdashline
    CSGLMM-g vs CSGLMMw-g & -1.371 & 0.059 \\
    \hdashline
    CSGEE-g vs CSGEEw-g & -1.362 & 0.053 \\
    \hdashline
    CSGEE-g vs CSGEEw-g (clus cov adj) & -2.040 & 0.500 \\
    \hline
\end{tabular}
}
\egroup
\end{center}
\end{table}

\newpage
\section{Extended case study re-analysis results}
\label{app:extended_case}

Complete and extended case study re-analysis results are included below as tables. This includes results from the appropriately weighted cluster-size saturated independence estimating equation with g-computation (CSIEE-g, CSIEEw-g).

\begin{table}[H]
\centering
\caption{Results from a re-analysis of the PPACT P-CRT are reported in terms of the point estimates, jackknife variance, and 95\% confidence intervals.}
\begin{tabular}{llrrr}
\toprule
 & Estimator & Estimate & Var & 95\% CI \\
\midrule
 & IEE       & -0.633 & 0.035 & (-1.004, -0.262) \\
 & CSIEE-g   & -0.646 & 0.036 & (-1.020, -0.272) \\
 & LMM       & -0.651 & 0.036 & (-1.026, -0.276) \\
 & LMM-MRS   & -0.632 & 0.035 & (-1.005, -0.260) \\
 & CSLMM-g   & -0.647 & 0.035 & (-1.020, -0.273) \\
\midrule
 & IEEw      & -0.702 & 0.041 & (-1.103, -0.301) \\
 & CSIEEw-g  & -0.727 & 0.041 & (-1.127, -0.327) \\
 & LMMw      & -0.726 & 0.046 & (-1.150, -0.301) \\
 & LMMw-MRS  & -0.702 & 0.041 & (-1.105, -0.299) \\
 & CSLMMw-g  & -0.728 & 0.041 & (-1.128, -0.328) \\
\bottomrule
\end{tabular}
\end{table}

\begin{table}[H]
\centering
\caption{The values of the test-statistic for detecting ICS and p-values with the different modeling approaches from a re-analysis of the PPACT P-CRT.}
\begin{tabular}{lrr}
\toprule
Comparison & Statistic & p-value \\
\midrule
IEEw vs.\ IEE             & -0.929 & 0.355 \\
CSIEEw-g vs.\ CSIEE-g & -1.061 & 0.291 \\
LMMw vs.\ LMM             & -0.889 & 0.376 \\
LMMw-MRS vs.\ LMM-MRS     & -0.922 & 0.359 \\
CSLMMw-g vs.\ CSLMM-g & -1.060 & 0.291 \\
\bottomrule
\end{tabular}
\end{table}


\end{document}